\pdfoutput=1

\documentclass[twocolumn,conference,10pt,letter]{IEEEtran}

\usepackage{multirow}
\usepackage{amsfonts,amsmath,amssymb,amsthm}
\usepackage{graphicx}
\usepackage{subcaption}
\usepackage{url}
\usepackage{color}

\newcommand{\be}{\begin{eqnarray}}
\newcommand{\ee}{\end{eqnarray}}

\newcommand{\ben}{\begin{enumerate}}
\newcommand{\een}{\end{enumerate}}
\newcommand{\beq}{\begin{equation}}
\newcommand{\eeq}{\end{equation}}
\newcommand{\beqa}{\begin{eqnarray*}}
\newcommand{\eeqa}{\end{eqnarray*}}
\newcommand{\bit}{\begin{itemize}}
\newcommand{\eit}{\end{itemize}}
\newcommand{\bt}{\begin{tabular}{c}}
\newcommand{\btt}{\begin{tabular}}
\newcommand{\et}{\end{tabular}}

\newtheorem{definition}{Definition}

\newtheorem{corollary}{Corollary}
\newtheorem{lemma}{Lemma}
\newtheorem{theorem}{Theorem}

\DeclareMathOperator*{\argmin}{arg\,min}
\DeclareMathOperator*{\argmax}{arg\,max}

\newcommand{\squishlist}{
   \begin{list}{$\bullet$}
    { \setlength{\itemsep}{0pt}      \setlength{\parsep}{0pt}
      \setlength{\topsep}{3pt}       \setlength{\partopsep}{0pt}
      \setlength{\listparindent}{-2pt}
      \setlength{\itemindent}{-5pt}
      \setlength{\leftmargin}{1em} \setlength{\labelwidth}{0em}
      \setlength{\labelsep}{0.5em} } }
\newcommand{\squishend}{
    \end{list}  }

\begin{document}




\title{Per-Server Dominant-Share Fairness (PS-DSF):\\ A Multi-Resource Fair Allocation Mechanism\\ for Heterogeneous Servers}

\author{
\IEEEauthorblockN{Jalal Khamse-Ashari\IEEEauthorrefmark{1}, Ioannis Lambadaris\IEEEauthorrefmark{1}, George Kesidis\IEEEauthorrefmark{2}, Bhuvan Urgaonkar\IEEEauthorrefmark{2} and Yiqiang Zhao\IEEEauthorrefmark{3}\\}
\IEEEauthorblockA{\IEEEauthorrefmark{1}Dept. of Systems and Computer Engineering, Carleton University, Ottawa, Canada\\
\IEEEauthorrefmark{2}School of EECS, Pennsylvania State University, State College, PA, USA\\
\IEEEauthorrefmark{3}School of Math and Statistics, Carleton University, Ottawa, Canada\\
Emails: \IEEEauthorrefmark{1}\{jalalkhamseashari,ioannis\}@sce.carleton.ca,~\IEEEauthorrefmark{2}\{gik2,buu1\}@psu.edu
~\IEEEauthorrefmark{3}zhao@math.carleton.ca}}

%

%
%
\maketitle

\begin{abstract}
Users of cloud computing platforms pose different types of demands for multiple resources on servers (physical or virtual machines). Besides differences in their resource capacities, servers may be additionally heterogeneous in their ability to service users - certain users' tasks may only be serviced by a subset of the servers. We identify important shortcomings in existing multi-resource fair allocation mechanisms - Dominant Resource Fairness (DRF) and its follow up work - when used in such environments. We develop a new fair allocation mechanism called Per-Server Dominant-Share Fairness (PS-DSF) which we show offers all desirable sharing properties that DRF is able to offer in the case of a single ``resource pool" (i.e., if the resources of all servers were pooled together into one hypothetical server). We evaluate the performance of PS-DSF through simulations. Our evaluation shows the enhanced efficiency of PS-DSF compared to the existing allocation mechanisms. We argue how our proposed allocation mechanism is applicable in cloud computing networks and especially large scale data-centers.
\end{abstract}

\section{Introduction}
Cloud computing has become increasingly popular as it provides a cost-effective alternative to proprietary high performance computing
systems. As the workloads to data-centers housing cloud computing platforms are intensively growing,
developing an efficient and fair allocation mechanism which guarantees quality-of-service for different workloads has become increasingly important.
Resource allocation and especially fair sharing in such shared computing system is particularly challenging because of the following reasons:
a) heterogeneity of servers, b) placement constraints, c) dealing with multiple types of resources, and d) diversity of workloads and demands. 

Real world data-centers are comprised of heterogeneous machines/servers with different configurations,
where some machines might be incompatible for some processing purposes/tasks.
Furthermore, each user may have specific requirements which further restrict the set of servers that the tasks of the user may run on.
For example, a user may require a machine with a public IP address, particular kernel version, special hardware such
as GPUs, or large amounts of memory, and might be unable to run on machines which lack such requirements.
For instance, it has been observed that over $50\%$ of tasks at Google clusters have strict constraints about the machines they can run on \cite{Mod_const, Ghodsi13}.

Besides  placement constraints, users present diversity over the amount of resources they need for executing one task.
For instance, the tasks of some users might be CPU intensive while for others memory or I/O bandwidth might be a bottleneck.
Dominant Resource Fairness is the first allocation mechanism which describes a notion of fairness when allocating multiple types of resources \cite{DRF}.
With DRF users receive a fair share of their \emph{dominant resource}.
Of all the resources requested by the user (for every unit of work called a task), its dominant resource is the one with the highest demand when expressed as a fraction of the overall resource capacity spread across all available servers. There are several other works investigating DRF allocation in case that different resources are distributed over heterogeneous servers but there are no placement constraints \cite{CDRF,Chiang12,HUG,DRFH15}.

There are some recent works investigating max-min fair allocation/scheduling for one type of resource while respecting placement constraints \cite{Ghodsi13,midrr,CM4FQ, Ashari16a,Ashari16b,Ashari17}. These schedulers could be useful in a multi-resource setting only when one of the resources serves as the bottleneck for all users, otherwise they might result in poor resource utilization \cite{DRF,Ghodsi13}. There are limited works in the literature investigating multi-resource fair allocation while respecting placement constraints \cite{UDRF, TSF, BLi16,scheduling-TR2}.
In this case, it is unclear how to globally identify the dominant resource as well as the dominant share for different users, as each user may have access only to a subset of servers. \cite{TSF,BLi16} present an elementary extension of DRF which identify the share of each user by ignoring the placement constraints and applying the same ideas as the un-constrained setting. We show that this approach does not achieve fairness even in the specific case that one of the resources serves as a bottleneck (Further discussions could be found in Section \ref{sec:main:challenge}).


\noindent {\bf Our Contributions:}
We propose a new allocation mechanism called \emph{Per-Server Dominant Share Fairness}.
We show that PS-DSF achieves all the desirable properties offered by DRF for a single resource pool: sharing incentive, strategy proofness, envy freeness, Pareto optimality, bottleneck fairness and single resource fairness (A detailed description of these properties can be found in Section \ref{sec:bg:drf}).
In fact, PS-DSF reduces to DRF when one considers a single server system.

The intuition behind PS-DSF is to compare and weigh the allocated resources to each user from the perspective of each server.
PS-DSF identifies a dominant resource and a virtual dominant share (VDS) for each user {\em with respect to each server} (as opposed to a single system-wide dominant share in DRF). The VDS for user $n$ with respect to (w.r.t.) server $i$ describes the fraction of the dominant resource which should be allocated to user $n$ from server $i$ as if all user $n$'s tasks were allocated resources solely from server $i$.
Each server may then use this localized metric to decide whether to increase or decrease
the allocated tasks to each user without the need to identify a global dominant share.
Besides its enhanced performance, 
PS-DSF is the first (to our knowledge) principled allocation mechanism which could be intrinsically implemented in a distributed manner.

The rest of this paper is organized as follows: In Section II, after describing the model, we give the necessary background and discuss insufficiency of the existing multi-resource allocation mechanisms, especially in case of heterogenous servers with placement constraints. After presenting our proposed allocation mechanism in Section III, we investigate different sharing properties that it satisfies and present a distributed algorithm to realize it.
We present some numerical experiments in Section IV, and finally we draw conclusions in Section V.

\section{Background and Model}\label{sec:bg}
Consider a set $\mathcal{K}$ of ${K}=|\mathcal{K}|$ heterogeneous servers (resource pools) each containing M types of resources.
We denote by $c_{i,r}$ the capacity of resource $r$ on server $i$, where $c_{i,r}\ge0$.
Let $\mathcal{N}$ denote the set of active users, where $N=|\mathcal{N}|$.
Let $d_n=[d_{n,r}]$ denote the per task \emph{demand vector} for user $n\in{\mathcal N}$,
that is the amount of each resource required for executing one task for user $n$.
Let $\phi_n>0$ denote the weight associated with user $n$. The weights reflect the priority of users with respect to each other.

Due to heterogeneity of users and servers, each user may be restricted to get service only from a subset of servers.
For example, each user may have some special hardware/software requirements (e.g., public IP address, a particular kernel version, GPU, etc.)
which restrict the set of servers that the tasks of the user may run on.
Besides such explicit placement constraints, users may not run their tasks on servers which lack some required resources.

For instance, consider the example in Figure \ref{fig:example1}, where three types of resources, CPU, memory, and network bandwidth are available over two servers in the amounts of ${\bf c}_1=[9\mbox{ cores}, 12\mbox{GB}, 100\mbox{Mb/s}]$ and ${\bf c}_2=[12 \mbox{ cores}, 12\mbox{GB}, 0\mbox{Mb/s}]$, where no communication bandwidth is available over the second server.
Consider three users with the weights $\phi_1=\phi_2=1,~\phi_3=2$, whose demand vectors are ${\bf d}_1=[1, 2, 10]$, ${\bf d}_2=[1, 2, 1]$ and ${\bf d}_3=[1, 2, 0]$.
Accordingly, users 1 and 2 are restricted to get service only from the first server, while user 3 may get service from both servers.
In summary, let $\delta_{n,i}=1$ if the tasks of user $n$ can run on server $i$, and otherwise $\delta_{n,i}=0$.

\begin{figure}[h!]
    \centering
\includegraphics[width=2.5in]{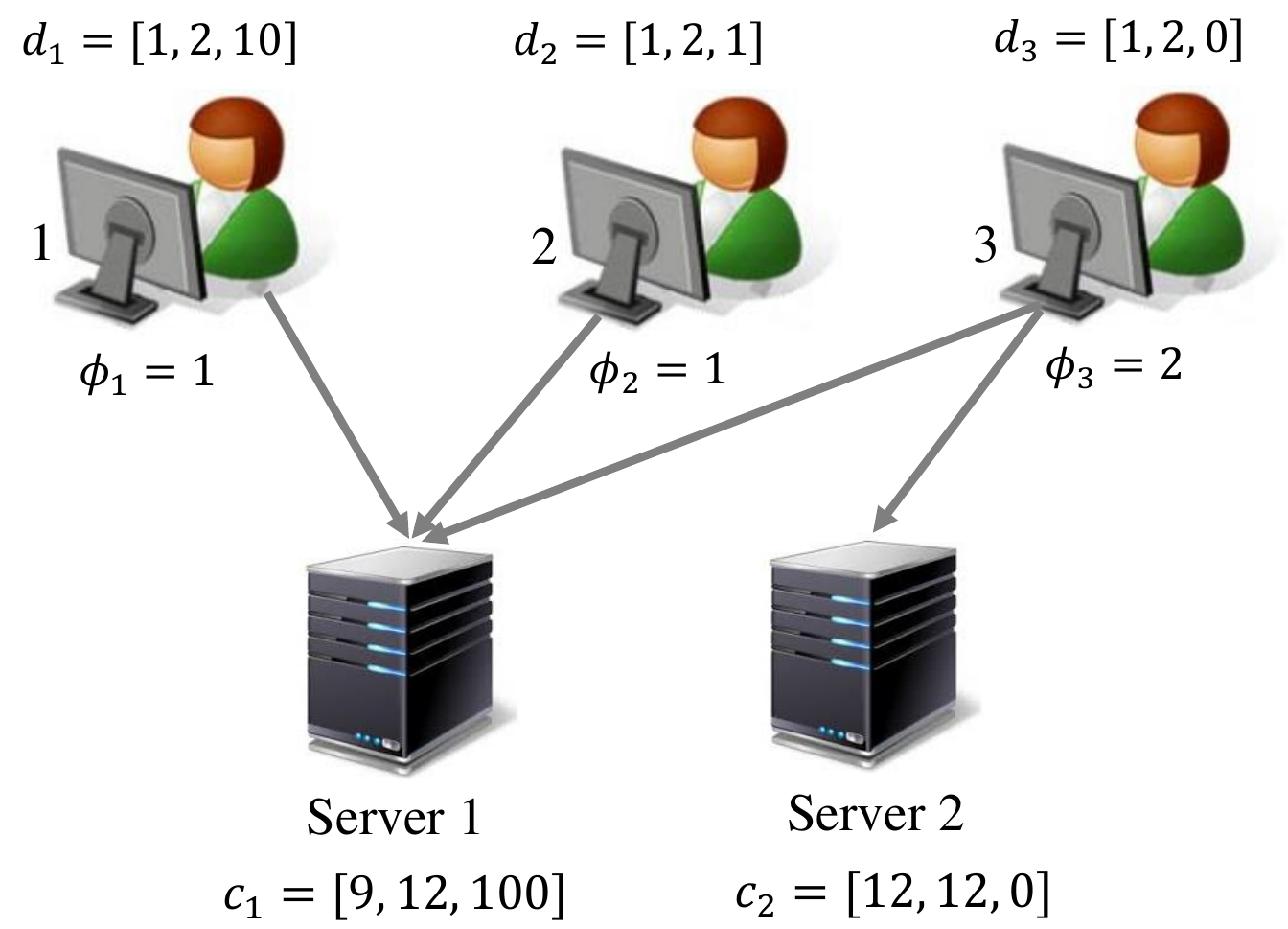}
\centering
\caption{\footnotesize A heterogeneous multi-resource system with three users and two servers.}
\label{fig:example1}
\vspace{-3mm}
\end{figure}

\subsection{Dominant Resource Fairness}\label{sec:bg:drf}
The problem of multi-resource fair allocation was originally studied in \cite{DRF} under the assumption that all resources are aggregated at one resource-pool. Specifically, let $c_r$ denote the total capacity of resource $r$.
Let ${\bf a}_n=[a_{n,r}]$ denote the amounts of different resources allocated to user $n$ under some allocation mechanism $\mathcal{A}$.
The utilization of user $n$ of its allocated resources, $U_n({\bf a}_n)$, is defined as the number of tasks,
$x_n$, which could be executed using ${\bf a}_n$, that is:
\be\label{utility}
U_n({\bf a}_n) := x_n = \min_r\frac{a_{n,r}}{d_{n,r}},
\ee
where, $x_n$ is a non-negative real number.
\cite{DRF} argues that the following important properties must be satisfied by a multi-resource allocation mechanism:
\squishlist
\item \emph{Sharing Incentive:} Consider a generic \emph{uniform allocation} where every user $n$ is allocated $\phi_n/\sum_m\phi_m$ portion of each resource. An allocation is said to provide sharing incentive, when each user is able to run more tasks compared to the uniform allocation.
\item \emph{Envy freeness:} A user should not prefer the allocation of another user when adjusted according to their weights, i.e., $U_n({\bf a}_n)\ge U_n(\frac{\phi_n}{\phi_m}{\bf a}_m)$, for all $m$.
\item \emph{Pareto Optimality:} It should not be possible to increase $x_n$ for any user $n$, without
    decreasing $x_m$ for some user $m$.
\item \emph{Strategy Proofness:} Users should not be able to increase their utilization by lying about their resource demands.
\squishend

%
%
%

Sharing incentive provides performance isolation, as it guarantees a minimum utilization for each user
irrespective of the demands of the other users. Envy freeness embodies the notion of fairness.
Pareto optimality results in maximizing system utilization. 
Finally, strategy proofness prevents users from gaming the allocation mechanism.
The reader is referred to \cite{DRF} or \cite{DRF12} for further details.

DRF is the first multi-resource allocation mechanism satisfying all the above properties.
Specifically, for every user $n$, the \emph{Dominant Resource} (DR) is defined as \cite{DRF}:
\be\label{DR}
\rho(n):=\argmax_rd_{n,r}/c_r,
\ee
that is, the resource whose greatest portion is required for execution of one task for user $n$.
The fraction of the DR that is allocated to user $n$ is defined as \emph{dominant share}:
\be\label{DS}
s_n:=\frac{a_{n,\rho(n)}}{c_{\rho(n)}}.
\ee

Without loss of generality, we may restrict ourselves to non-wasteful allocations, i.e., ${\bf a}_n=x_n{\bf d}_n,~\forall n$.
In this case, an allocation $\{x_n\}$ is feasible when:
\begin{eqnarray}
&&\sum_nx_nd_{n,r}\le c_r,~\forall r.
\end{eqnarray}
\begin{definition}
It is said that $\{x_n\}$ satisfies \emph{DRF}, if it is feasible and the normalized dominant share for each user, $s_n/\phi_n$ cannot be increased while
maintaining feasibility without decreasing $s_m$ for some user $m$ with $s_m/\phi_m\le s_n/\phi_n$ \cite{DRF}.
\end{definition}

DRF is a restatement of max-min fairness in terms of \emph{dominant shares}.
What make it appealing are desirable sharing properties which are satisfied under this allocation mechanism.
Besides the above-mentioned essential properties, DRF also satisfies the following simple but essential properties \cite{DRF}.

\squishlist
\item \emph{Single Resource Fairness:} When there is only one resource type, the allocation satisfies max-min fairness.
\item \emph{Bottleneck Fairness:} If there is one resource which is dominantly requested by each user, 
then the allocation satisfies max-min fairness for that resource.
\squishend

\subsection{Challenges with Heterogeneous Resource-Pools and Placement Constraints}\label{sec:main:challenge}
The notion of DRF has been extended to the case of heterogenous servers, when 
all types of resources are available within each server and there are no placement constraints \cite{DRFH15}.
In this case, DR for user $n$ is readily identified as the resource whose greatest portion is required for
execution of one task as if all resources were integrated at resource pool.
That is, DR for user $n$ could be identified according to \eqref{DR},
where $c_r:=\sum_ic_{i,r}$ is the total capacity of resource $r$.
Furthermore, the \emph{global dominant share} for user $n$ is given by:
\be
s_n=x_n\max_r\frac{d_{n,r}}{c_{r}},
\ee
where $x_n$ here is the total number of tasks which are allocated to user $n$ from different servers, that is $x_n:=\sum_ix_{n,i}$.
In \cite{DRFH15} it is proposed to find $\{x_{n,i}\}$ such that max-min fairness is achieved in terms of global dominant shares.
This mechanism, which is referred to as DRFH, has been shown to achieve Pareto optimality, strategy proofness, envy freeness and bottleneck fairness.
However, it fails to provide sharing incentive \cite{DRFH15}.

When there are placement constraints, it is unclear how to define a single system-wide DR for a user similar to that in \cite{DRFH15}.
A natural first thought may be to identify the DR over the set of eligible servers for each user.
However, in this case users may have an incentive to misreport the set of eligible servers \cite{TSF}.
A strategy-proof approach is to identify the DR for each user as
\emph{if} there were no placement constraints and all resources were integrated at one resource pool.
We argue that this approach, which we refer to as C-DRFH, does not result in a fair allocation as it does not satisfy bottleneck fairness.

To appreciate this shortcoming of C-DRFH, consider the example in Figure \ref{fig:example1},
where the second resource (RAM) is dominantly requested by every user from its eligible servers.
If we allocate the available RAM proportionate to the weights,
6GB is allocated to the first two suers and 12 GB is allocated to the third user.
Accordingly, each user is allocated $x_1=x_{1,1}=3$, $x_2=x_{2,1}=3$ $x_3=x_{3,2}=6$ tasks
(this allocation follows from our proposed allocation mechanism - see Section~\ref{sec:main}).
However, C-DRFH would instead identify bandwidth as the dominant resources for the first user and identifies RAM as the dominant resource for the second and third users. Hence, if we allocate global dominant shares in a weighted fair manner, each user is allocated $x_1=x_{1,1}=2.609$, $x_2=x_{2,1}=3.130$,
and $x_3=x_{3,1}+x_{3,2}=6+0.261=6.261$ tasks respectively, which obviously violates fairness on the bottleneck resource.


\vspace{-1mm}
\begin{table}[ht]
\footnotesize
\centering
\captionsetup{name=Table }
\caption{\footnotesize Properties of different allocation mechanisms in case of heterogeneous servers with placement constraints: sharing~incentive (SI), envy~freeness (EF), strategy~proofness (SP), Pareto optimality (PO), and bottleneck fairness (BF).}
\label{table2}
\begin{tabular}{|c||c|c|c|}
\hline
Property      &      C-DRFH         &     ~~TSF~~        &       PS-DSF\\
\hline
\hline
SI            &                     &  $\checkmark$  &      $\checkmark$\\
\hline
EF            &   $\checkmark$      &  $\checkmark$  &      $\checkmark$\\
\hline
SP            &   $\checkmark$      &  $\checkmark$  &      *\\
\hline
PO            &   $\checkmark$      &  $\checkmark$  &      *\\
\hline
BF            &                     &                &      $\checkmark$\\
\hline
\end{tabular}
\end{table}

Yet another extension of DRF that also considers heterogeneous servers, all containing all types of resources without any placement constraints, is CDRF~\cite{CDRF}.
Specifically, let $\gamma_n:=\sum_i\gamma_{n,i}$ be defined as the number of tasks which are allocated to user $n$
when monopolizing the whole cluster (i.e., if $n$ were the only user running on the cluster).
An allocation is said to satisfy CDRF\footnote{Containerized DRF}, when $x_n/\gamma_n$ satisfies max-min fairness.
In case of one server, $x_n/\gamma_n$ gives dominant share for each user $n$.
As a result, CDRF reduces to DRF in case of one server. In case of multiple heterogeneous servers with no placement constraints,
 CDRF is shown to satisfy Pareto optimality, strategy proofness, envy freeness and sharing incentive properties \cite{CDRF}.

In \cite{TSF} CDRF has been extended to address the placement constraints.
Specifically, let $\gamma_n:=\sum_i\gamma_{n,i}$ be (re)defined as the number of tasks which are allocated to user $n$
from different servers when monopolizing all servers as \emph{if} there were no placement constraints \cite{TSF}.
An allocation is said to satisfy Task Share Fairness (TSF), when $x_n/\gamma_n$ satisfies max-min fairness.
TSF is shown to satisfy Pareto optimality, strategy proofness, envy freeness and sharing incentive properties
in case of heterogeneous servers with placement constraints \cite{TSF}.
However, we argue that this mechanism is not essentially fair as it does not satisfy bottleneck fairness.

For instance, consider again the example in Figure \ref{fig:example1}.
The number of tasks that each user may run in the whole cluster is $\gamma_1=\gamma_2=6$,  and $\gamma_3=12$ tasks, respectively.
Hence, each user is allocated $x_1=x_{1,1}=2$, $x_2=x_{2,1}=2$ and $x_3=x_{3,1}+x_{3,2}=6+2=8$ tasks according to TSF mechanism,
which is completely different and far from the fair allocation.
Table \ref{table2} summarizes different sharing properties which could be satisfied under different allocation mechanisms.
Shortcomings of the existing allocation mechanisms in case of heterogeneous servers with placement constraints motivates us to develop a new allocation mechanism.

\section{Per-Server Dominant Share Fairness}\label{sec:main}
In this section we describe PS-DSF, an extension of DRF that is applicable for heterogeneous resource-pools in the presence of placement constraints.
As discussed in the previous section, in the case of heterogeneous servers and in the presence of placement constraints, it is unclear how to globally identify one DR and the corresponding dominant share for each user. 
The intuition behind PS-DSF is to define a \emph{virtual dominant share} for every user w.r.t. each server.
Towards this, we first define the DR for every user $n$ w.r.t. each server $i$ as:
\be
\rho(n,i):=\argmax_r\frac{d_{n,r}}{c_{i,r}}.
\ee
Let $\gamma_{n,i}$ denote the number of tasks which could be executed by user $n$
when monopolizing server $i$:
\be
\gamma_{n,i}:=\delta_{n,i}\min_r\frac{c_{i,r}}{d_{n,r}}=\delta_{n,i}\frac{c_{i,\rho(n,i)}}{d_{n,\rho(n,i)}}.
\ee
We say that server $i$ is \emph{eligible} to serve user $n$ when $\gamma_{n,i}>0$ or equivalently $\delta_{n,i}=1$.
Without loss of generality we restrict ourselves to non-wasteful allocations, that is ${\bf a}_{n,i}=x_{n,i}{\bf d}_n$, where ${\bf a}_{n,i}=[a_{n,i,r}]$ is the vector of allocated resources to user $n$ from server $i$ and $x_{n,i}\in\mathbb{R}^+$ is the number of allocated tasks from the same server.
\begin{definition}
The Virtual Dominant Share (VDS) for user $n$ w.r.t. server $i$, $s_{n,i}$, is defined as:
\be
s_{n,i}=\frac{x_n}{\gamma_{n,i}},
\ee
where $x_n=\sum_jx_{n,j}$ is the \emph{total number of tasks} that are allocated to user $n$
(whether or not these tasks are actually allocated using server $i$).
\end{definition}
Intuitively, $s_{n,i}$ gives the fraction\footnote{The reader may note that $s_{n,i}$ could be possibly greater than 1, as some tasks could be allocated to user $n$ from other servers.} of the dominant resource for user $n$ w.r.t. server $i$ which should be allocated to it as if $x_n$ tasks were allocated to it entirely from server $i$.
When the available resources over each server are arbitrarily divisible,
we have the following condition on $\{x_{n,i}\}$ to be feasible.
\begin{definition} An allocation, $\{x_{n,i}\}$, is said to satisfy Resource Division Multiplexing (RDM) constraint, when:
\begin{eqnarray}
&&\sum_nx_{n,i}d_{n,r}\le c_{i,r},~\forall i,r.\label{FC1_1}
\end{eqnarray}
\end{definition}
For a data-center comprising of a plurality of servers, it is sometimes of more practical interest to assume that servers \emph{may not} be divided to finer partitions~\cite{Ghodsi13}. Accordingly, the hypervisor may only time-share servers among different users.
In this case, we have the following condition on $\{x_{n,i}\}$ to be feasible.
\begin{definition}
An allocation, $\{x_{n,i}\}$, is said to satisfy Time Division Multiplexing (TDM) constraint, when\footnote{Considering resources such as CPU, BW, $\cdots$ , which are attributed a processing speed per time-unit, $x_{n,i}/\gamma_{n,i}$ represent the percentage of time-unit that server $i$ is allocated to user $n$.}:
\begin{eqnarray}
&&\sum_nx_{n,i}/\gamma_{n,i}\le 1,~\forall i.\label{FC2_1}
\end{eqnarray}
\end{definition}
It can be observed that TDM constraint is more stringent than RDM constraint, as \eqref{FC2_1} implies \eqref{FC1_1}:
\be
1 \ge \sum_n\frac{x_{n,i}}{\gamma_{n,i}}=\sum_n\frac{x_{n,i}d_{n,\rho(n,i)}}{c_{i,\rho(n,i)}}\ge\frac{\sum_nx_{n,i}d_{n,r}}{c_{i,r}},~\forall i,r.
\ee
We investigate our proposed allocation mechanism under both of these feasibility conditions.

\begin{definition}\label{PS_DSF_Def}
An allocation $\{x_{n,i}\}$ satisfies \emph{Per-Server Dominant-Share Fairness}, if it is feasible and the allocated tasks to each user, $x_n$ cannot be increased while maintaining feasibility without decreasing $x_{m,i}$ for some user $m$ and server $i$ with $s_{m,i}/\phi_m\le s_{n,i}/\phi_n$.
\end{definition}

\subsection{An Example}
Consider again the heterogeneous servers from our earlier example this time serving four equally weighted users whose demand vectors
are ${\bf d}_1=[1.5,1,10]$, ${\bf d}_2=[1,2,10]$, ${\bf d}_3=[0.5,1,0]$, ${\bf d}_4=[1,0.5,0]$. We show this in Figure~\ref{fig:example21}.
Note the placement constraints for users 1 and 2 whose tasks may only run on the first server.
We show the PS-DSF allocation (based on RDM) in Figure~\ref{fig:example22}.
The allocated tasks to each user are $x_1=x_{1,1}=3.6$, $x_2=x_{2,1}=3.6$, $x_3=x_{3,2}=8$, $x_4=x_{4,2}=8$, respectively,
where no tasks are allocated to users 3 and 4 from the first server. Specifically, the VDS for user 3 (and user 4 respectively) w.r.t. the first server
is $s_{3,1}=8/12$ ($s_{4,1}=12/12=1$), while the VDS of users 1 and 2 w.r.t. this server is $s_{1,1}=s_{2,1}=0.6$.
The VDS of users 3 and 4 w.r.t. the second server is $s_{3,1}=s_{4,1}=8/12$. The reader may verify that for each server $i$ the allocated tasks to each user may not be increased without decreasing the allocated tasks of some user with less VDS.

\begin{figure}[h!]
    \centering
\includegraphics[width=3in]{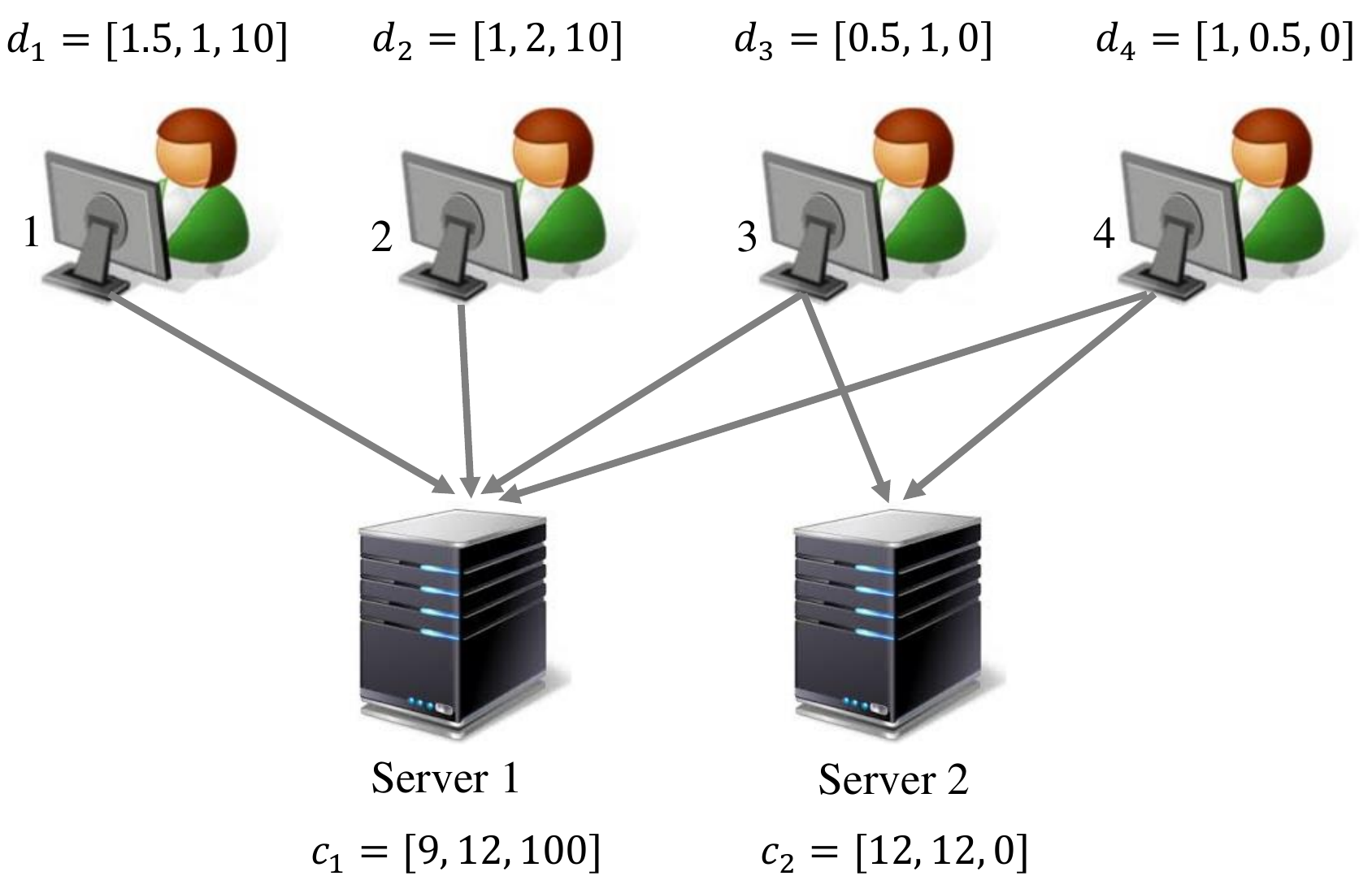}
\centering
\footnotesize
\caption{\footnotesize A heterogeneous multi-resource system with four users and two servers.}
\label{fig:example21}
\end{figure}
\begin{figure}[h!]
    \centering
\includegraphics[width=3.5in]{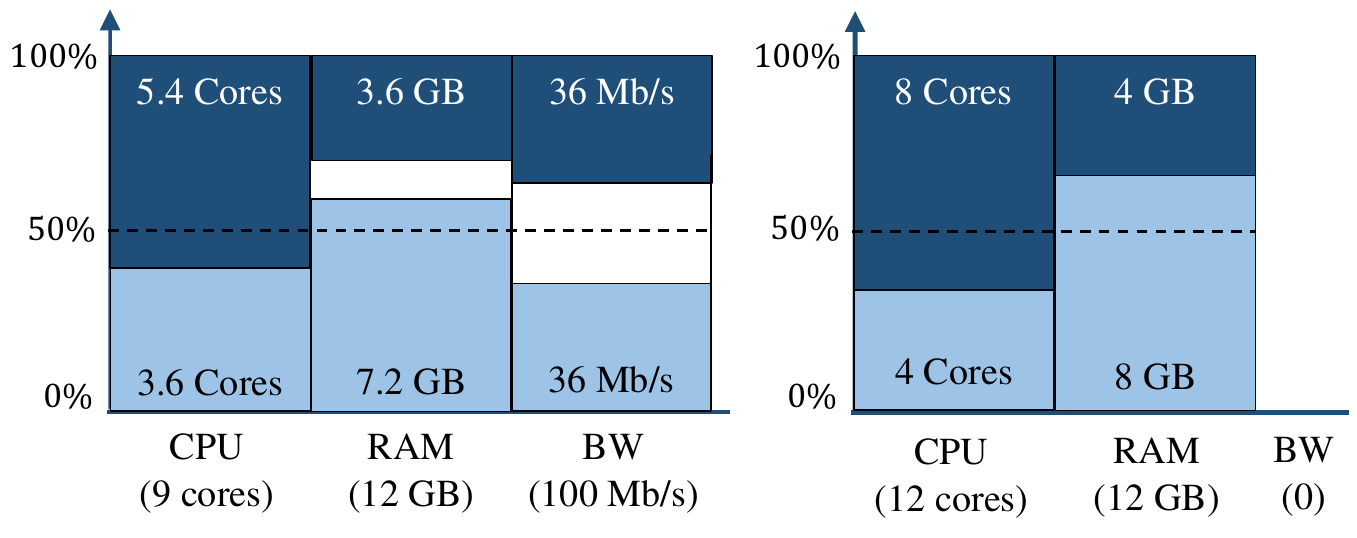}
\centering
\caption{\footnotesize PS-DSF allocation for the example in Figure \ref{fig:example21}.}
\label{fig:example22}
\end{figure}

\subsection{The Properties of the PS-DSF Allocation Mechanism}\label{sec:main:properties}
Before examining different sharing properties satisfied under PS-DSF allocation mechanism,
we describe a necessary and sufficient condition to achieve PS-DSF.

\begin{definition}\label{Def_bottleneck}
Given a feasible allocation $\{x_{n,i}\}$ based on RDM, we say that $r$ is a \emph{bottleneck} resource for user $n$ w.r.t. an eligible server $i$ if $d_{n,r}>0$, $\sum_mx_{m,i}d_{m,r}=c_{i,r}$ (i.e. $r$ is \emph{saturated}), and
\be
\frac{s_{n,i}}{\phi_n}\ge\frac{s_{m,i}}{\phi_m},~\forall m\mbox{ such that }x_{m,i}d_{m,r}>0.
\ee
\end{definition}

\begin{theorem}\label{NS_condition}
A feasible allocation $\{x_{n,i}\}$ based on RDM satisfies PS-DSF if and only if there exists a bottleneck resource for every user w.r.t. every eligible server.
\end{theorem}

\begin{theorem}\label{NS_condition2}
A feasible allocation $\{x_{n,i}\}$ based on TDM satisfies PS-DSF if and only if \eqref{FC2_1} holds with equality, and
\be
\frac{s_{n,i}}{\phi_n}\ge\frac{s_{m,i}}{\phi_m},~\forall n\mbox{ and }\forall m\mbox{ such that }x_{m,i}>0.\label{Per_server_cond2}
\ee
\end{theorem}


The proofs are given in the appendix.
These conditions will be useful in determining a PS-DSF allocation (see Section \ref{sec:PS-DFS:algo}).
In the following we examine different sharing properties that are satisfied under PS-DSF.
In case of a heterogeneous system with placement constraints,
we will need to extend the notion of \emph{Sharing Incentive}, \emph{Strategy Proofness} and \emph{Bottleneck Fairness}. Other properties,
\emph{Envy Freeness}, \emph{Pareto Optimality} and \emph{Single Resource Fairness}, will follow the same definitions as described in Section \ref{sec:bg:drf}.

The generalization of sharing incentive property is straightforward. We consider a \emph{uniform allocation} which allocates $\phi_n/\sum_m\phi_m$ portion of the resources on each server (whether this server is eligible or not) to each user $n$.
An allocation is said to satisfy \emph{sharing incentive}, when each user is able to run more tasks compared to such uniform allocation.

For the strategy proofness property, we may note that \emph{we assume} each user to declare its demand vector and also the set of eligible servers. We say that an allocation satisfies \emph{strategy proofness} when users may not increase their utilization by lying about their resource demands or the set of eligible servers.

Finally, a resource is considered as a bottleneck in the whole system
when it is dominantly requested by each user from every eligible server.
If there is a bottleneck resource, then the allocation should satisfy max-min fairness w.r.t. such resource.

\begin{theorem}\label{Properties}
PS-DSF allocation mechanism (whether based on RDM or TDM) satisfies single resource fairness, bottleneck fairness,
envy freeness, and sharing incentive properties. It also satisfies Pareto optimality and strategy proofness in case of TDM.
\end{theorem}

The proof is given in the appendix.
Unfortunately PS-DSF does not satisfy Pareto optimality in case of RDM.
This is the reason why strategy proofness is \emph{not generally satisfied} in case of RDM.
The following lemma describes the behaviour of PS-DSF allocation mechanism from this respect.
\begin{lemma}\label{strategy_proofness}
Assume that all users demand all type of resources, that is $d_{n,r}>0,~\forall n,r$.
Under the PS-DSF allocation mechanism with RDM,
each user cannot decrease the utilization of other users by lying about its resource demands or the set of eligible
servers, without decreasing its own utilization.
\end{lemma}

For the proof refer to the appendix.


\subsection{PS-DSF Allocation Algorithm}\label{sec:PS-DFS:algo}
In this subsection, we present an algorithm which realizes the PS-DSF allocation
in the case of RDM\footnote{A simplified version of this algorithm can be used in the case of TDM.}.
According to Theorem~\ref{NS_condition}, an allocation satisfies PS-DSF when every user has a bottleneck resource w.r.t. every eligible server.
Let $\mathcal{N}_i$ denote the set of users for which $\gamma_{n,i}>0$. The following corollary describes a condition to check whether a saturated resource serves as a bottleneck for user $n$ w.r.t. server $i$.

\begin{corollary}\label{corollary_1}
If $r$ is saturated at server $i$ and
\be
n\in\argmin_{m\in\mathcal{N}_i}\{\frac{s_{m,i}}{\phi_m}\mid d_{m,r}>0\}
\ee
Then, $r$ is a bottleneck resource for user $n$ at server $i$ when:
\be
\frac{s_{m,i}}{\phi_m}>\frac{s_{n,i}}{\phi_n},~d_{m,r}>0~\Rightarrow~x_{m,i}=0.\label{ch_cond}
\ee
\end{corollary}

To find a PS-DSF allocation, we may apply an iterative algorithm beginning with an initial allocation.
Assume that servers are indexed from 1 to $K$. Starting from the first server,
the proposed algorithm sequentially updates the allocation for different servers,
so that at the end a bottleneck resource is identified for every user w.r.t. every eligible server.
In the following we describe the procedure for updating the allocation at each server.

Specifically, for each server $i$ let $\mathcal{N}_i$ initially denote the set of users for which $\gamma_{n,i}>0$.
Given a feasible allocation, $\{x_{n,i}\}$, find $S_i^*$ as the \emph{minimum VDS} at server $i$:
\be
S_i^*:=\min_{m\in\mathcal{N}_i}\{\frac{s_{m,i}}{\phi_m}\}.\label{Fst_lvl}
\ee
The set of users achieving the minimum in \eqref{Fst_lvl} is denoted by $\mathcal{N}_i^*$.
Let $\mathcal{R}_i^*$ denote the set of saturated resources at server $i$ for which $d_{n,r}>0$ for some user $n\in\mathcal{N}_i^*$.
These resources are the potential bottleneck resources for users $n\in\mathcal{N}_i^*$.
If the condition in Corollary~\ref{corollary_1} is satisfied for some resource $r^*\in\mathcal{R}_i^*$,
then this resource serves as the bottleneck for users $n\in\mathcal{N}_i$ with $d_{n,r^*}>0$.
In this case, we restrict our attention to the users for which no bottleneck resource is identified w.r.t. server $i$.
Specifically, $\mathcal{N}_i$ is updated to:
\be
\mathcal{N}_i=\mathcal{N}_i-\{n\mid d_{n,r^*}>0\}.
\ee

When the condition in Corollary~\ref{corollary_1} is not satisfied for any resource $r\in\mathcal{R}_i^*$,
the algorithm updates the allocation for server $i$.
Specifically, for every resource $r\in\mathcal{R}_i^*$, a user $n_r$ is chosen such that:
\be
n_r\in\argmax_{n\in\mathcal{N}_i}\{\frac{s_{n,i}}{\phi_n}\mid x_{n,i}d_{n,r}>0\}.
\ee
If we release the whole allocated resources to these users from server $i$,
the maximum potential increase in $S_i^*$ is given by $z^*$ (see the \emph{Update-Allocation} subroutine in Algorithm \ref{table1}).
To make sure that $S_i^*$ is monotonically increasing, $\beta\in(0,1]$ is chosen such that $S_i^*+\beta z^*$ remains
less than or equal to the updated VDS w.r.t. server $i$ for all users $n_r,~r\in\mathcal{R}_i^*$.


For each server $i$, the above procedure is repeated until $\mathcal{N}_i$ becomes empty.
At the end of this procedure, a bottleneck resource is identified for every user eligible to be served by server $i$.
However, the subsequent updates for the next servers, may violate this condition for server $i$ and the previous servers.
Hence, we repeat the whole process for all servers, until no more update is possible for any of the servers\footnote{Convergence properties of this algorithm will be studied in our future work.}.
This process is described in Algorithm I.
%

\begin{table}
\captionsetup{labelformat=empty}
\caption{Algorithm I: PS-DSF Allocation Algorithm}
\label{table1}
\begin{tabular}{p{8.48cm}}
\hline\noalign{\smallskip}
\squishlist
\item[] {\bf Initialization}
\squishend
$\quad$ Initially allocate available resources by applying DRF individually to each server.

\squishlist
\item[] {\bf The main subroutine}
\squishend
$\quad$ {\bf while }($1$)\\
$\quad$ $\qquad$Last-round-flag$~:=1$\\
$\quad$ $\qquad${\bf for }$(i=1;~i\le K;~i++)$\\
$\quad$ $\qquad\qquad\mathcal{N}_i:=\{n\in\mathcal{N}\mid\gamma_{n,i}>0\}$.\\
$\quad$ $\qquad\qquad${\bf while }($\mathcal{N}_i\neq\emptyset$)\\
$\quad$ $\qquad\qquad\qquad$Find $S_i^*$ according to \eqref{Fst_lvl}.\\
$\quad$ $\qquad\qquad\qquad$Identify $\mathcal{N}_i^*$ as the set of users achieving the minimum in \eqref{Fst_lvl}.\\
$\quad$ $\qquad\qquad\qquad$Identify $\mathcal{R}_i^*$ as the set of saturated resources at server $i$ for which\\
$\quad$ $\qquad\qquad\qquad$$d_{n,r}>0$ for some $n\in\mathcal{N}_i^*$.\\
$\quad$ $\qquad\qquad\qquad${\bf If }$(S_i^*=\max_{n\in\mathcal{N}_i}\{\frac{s_{n,i}}{\phi_n}\mid x_{n,i}d_{n,r^*}>0\},$~for $r^*\in\mathcal{R}_i^*)$\\
$\quad$ $\qquad\qquad\qquad\qquad$Update $\mathcal{N}_i=\mathcal{N}_i-\{n\mid d_{n,r^*}>0\}$\\
$\quad$ $\qquad\qquad\qquad${\bf else }\\
$\quad$ $\qquad\qquad\qquad\qquad$Last-round-flag$~=0$\\
$\quad$ $\qquad\qquad\qquad\qquad$Call \emph{Update-Allocation$({\bf x},i)$}.\\
$\quad$ $\qquad${\bf If }(Last-round-flag$~=1$)\\
$\quad$ $\qquad\qquad$break;

\squishlist
\item[]{\bf Update-Allocation$({\bf x},i)$ subroutine}
\squishend

$\quad$Identify ${\bf f}_i=[f_{i,r}]$ as the amount of unallocated resources under ${\bf x}$.\\
$\quad${\bf for} ($r\in\mathcal{R}_i^*$)\\
$\quad$ $\qquad$Choose $n_r\in\argmax_{n\in\mathcal{N}_i}\{\frac{s_{n,i}}{\phi_n}\mid x_{n,i}d_{n,r}>0\}.$\\
$\quad$ $\qquad$Update ${\bf f}_i={\bf f}_i + x_{n_r,i}{\bf d}_{n_r}$.\\\vspace{+1mm}
$\quad$Find $D_i^*:=\sum_{n\in\mathcal{N}_i^*}\phi_n\gamma_{n,i}{\bf d}_n$.\\
$\quad$Find $z^*:=\min_r\frac{f_{i,r}}{D^*_{i,r}}$.\\
$\quad$Choose $\beta\in(0,1]$ such that: $S_i^*+\beta z^*\le\frac{x_{n_r}-\beta x_{n_r,i}}{\phi_{n_r}\gamma_{n_r,i}},~\forall r\in\mathcal{R}_i^*.$\\\vspace{+0.1mm}
$\quad$Update $x_{n,i}=x_{n,i}+\beta\phi_n\gamma_{n,i}z^*,~\forall n\in\mathcal{N}_i^*$.\\\vspace{+0.1mm}
$\quad$Set $x_{n_r,i}=(1-\beta)x_{n_r,i}~\forall r\in\mathcal{R}_i^*$.\\
\noalign{\smallskip}\hline
\end{tabular}
\end{table}

\subsection{Distributed Implementation}\label{sec:main:DI}

One of the advantages of the PS-DSF allocation mechanism is that it locally identifies the dominant resource for every user w.r.t. each server,
without any knowledge of the available resources on the other servers, as opposed to existing allocation mechanisms which
need to \emph{globally} identify dominant resource and/or dominant share for each user.
This is of great importance from a practical point of view, as we may develop a distributed algorithm to find the PS-DSF allocation.

Specifically, consider the inner while-loop in the main sub-routine of Algorithm I which we refer to as \emph{``server procedure"}.
According to this procedure the allocated tasks to different users from each server $i$ are updated only based on the knowledge of the available resources on server $i$ and the total allocated tasks to each user. Accordingly, we may come up with a distributed version of Algorithm I where each server individually (and even asynchronously) executes the server procedure every $T$ seconds.
When $T$ is chosen sufficiently smaller than period of changes in a cluster (like changes in the set of active users and/or servers),
such distributed algorithm may dynamically achieve the PS-DSF allocation. We implement this algorithm in our experiments in Section~\ref{sec:eval}.

\section{Extensions}
In this section we present an extension which directly follows from our proposed approach in Section \ref{sec:main}.
Specifically, consider the case where the effective capacity of resources on each server may vary for different users.
In this case, that is unclear how to define a global dominant resource for each user 
even when there are no placement constraints. However, according to the formulation in Section \ref{sec:main}, we may define $\gamma_{n,i}$ as the number of tasks which could be executed by user $n$ when monopolizing server $i$. We may also define the VDS for every user w.r.t. each server in the same way, and then find an allocation which satisfies PS-DSF. To gain more intuition, we consider two specific example scenarios in the following.

{\bf Example scenario 1:}
First consider a simple scenario where only one type of resource, notably bandwidth, is available on different servers.
Specifically, we may consider different servers as different frequency channels which are subject to multi-user diversity in a wireless system.
For instance, consider the example in Figure~\ref{fig:example3} where two users share three wireless channels.
Without the insight of our proposed approach, that is unclear how to allocate the capacity of servers among different users in a fair manner,
as we may not weigh different servers with respect to each other.

Let define the utility of each user, $x_n$, as the number of bits which are given service in one second (i.e. the service rate).
Also define $\gamma_{n,i}$ as the achievable service rate by user $n$ when monopolizing server $i$.
For the example in Figure~\ref{fig:example3} PS-DSF results in allocating the first channel (the third channel respectively) to the first user (second user), while the second channel is equally shared between the two users. Accordingly, user 1 gets a service rate of 1.5Mb/s while user 2 gets 1Mb/s.
It can be observed that $x_{n}$ can not be increased for any user without decreasing $x_{m,i}$ while $x_m/\gamma_{m,i}\le x_n\gamma_{n,i}$.

\begin{figure}[h!]
    \centering
\includegraphics[width=2.5in]{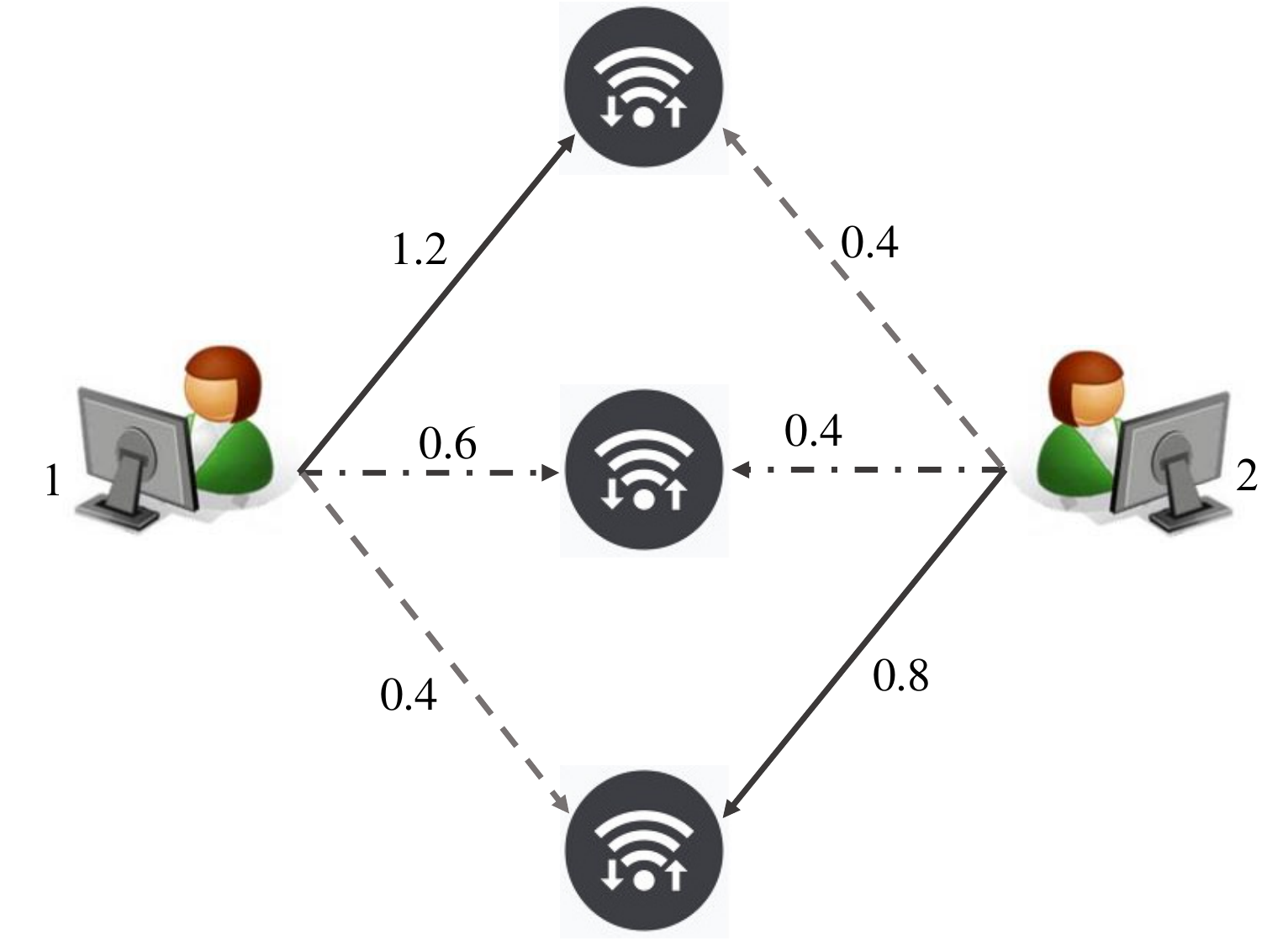}
\centering
\caption{An example with two equally weighted users sharing three frequency channels.
The achievable service rates by each user over different channels are shown beside the arrows (in Mb/s).}
\label{fig:example3}
\end{figure}
\vspace{-3mm}

{\bf Example scenario 2:}
Consider a set of heterogeneous servers in a computing cluster, where each server consists of different types of resources, such as CPU, RAM, bandwidth, etc.
Although the CPU on each server has a fixed physical capacity, different users may experience different effective processing capacities when specific  co-processors are available at a server.
Coprocessors are supplementary processing units which are specialized for specific arithmetics or other processing purposes.
As a result, they might be useful only for some users, for which they accelerate processing performance.

Our approach in Section \ref{sec:main} could be readily extended to incorporate the effect of coprocessors.
Assume that ${\bf d}_n=[d_{n,r}]$ describes the demand of user $n$ from each type of resource when no co-processor is utilized.
Let $\gamma_{n,i}$ denote the maximum number of tasks which could be executed by user $n$ when monopolizing server $i$ and utilizing any available co-processor. Given the insight of PS-DSF, we may find an allocation, $\{x_{n,i}\}$, such that for every user, $x_n$ cannot be increased while maintaining feasibility without decreasing $x_{m,i}$ for some user $m$ and sever $i$ with $x_m/\phi_m\gamma_{m,i}\le x_n/\phi_n\gamma_{n,i}$.
This problem will be studied in more details in our future work.

\section{Numerical Results}\label{sec:eval}
In this section we evaluate performance of the PS-DSF allocation mechanism through some numerical experiments.
In our simulations, we consider a cluster with four different classes of servers (120 servers in total), where the configuration of servers are
drawn from the distribution of Google cluster servers \cite{Google11}.
It is assumed that the available resources over each server can be partitioned in any arbitrary way.
We consider four users where the last two users may run their tasks only by the last two classes of servers (see Figure \ref{fig:num_exmp}).

\begin{figure}[h!]
\centering
\includegraphics[width=2.3in]{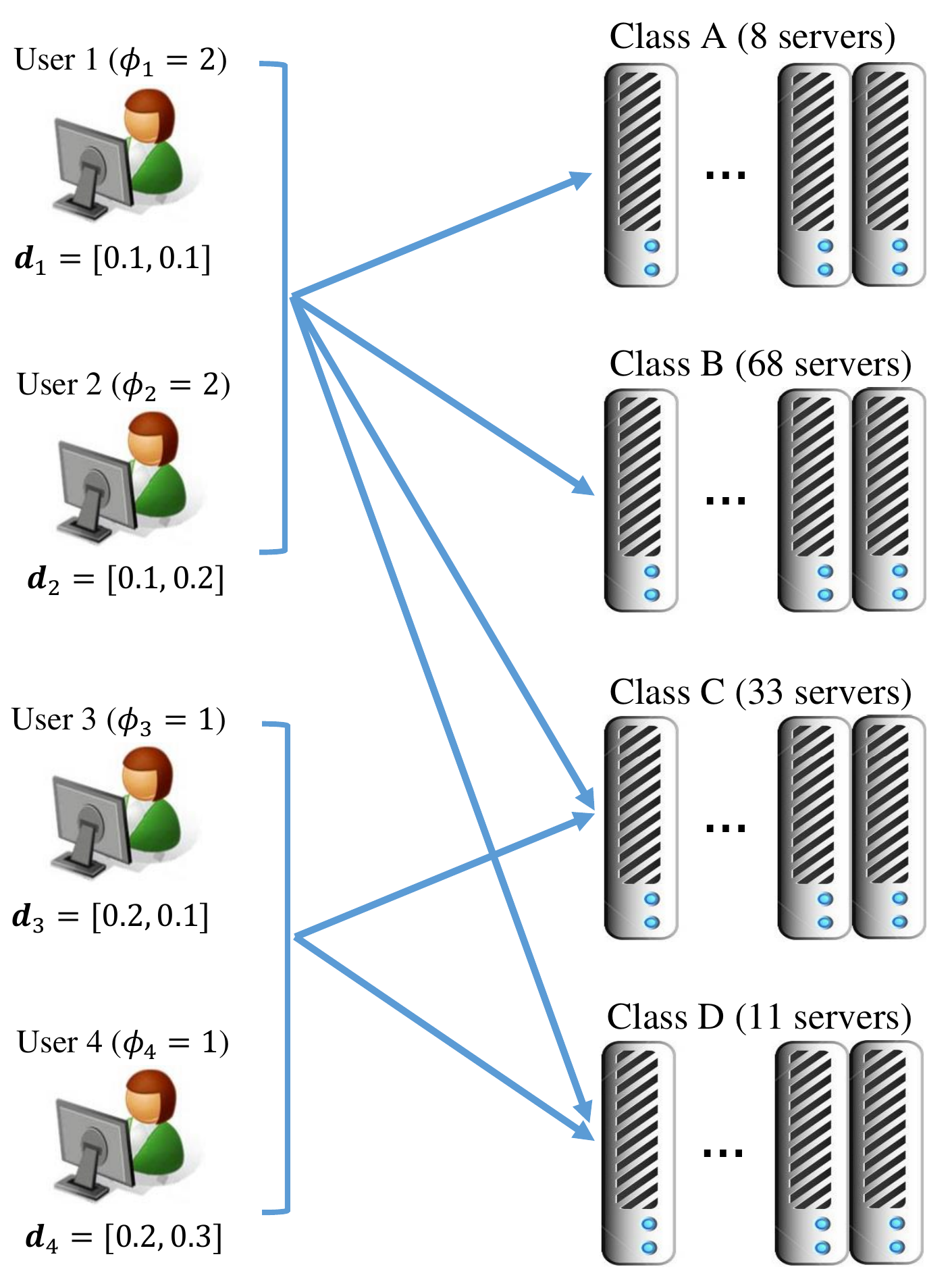}
\footnotesize
\caption{\footnotesize A cluster with four classes of servers (120 servers in total) and four users. The weight of the first two users is twice the weight of the last two users. The configurations of resources (CPU and memory respectively) for servers of each class are as follows: $C_A=[1,1],~C_B=[0.5,0.5],~C_C=[0.5,0.25],~C_D=[0.5,0.75]$, where CPU and memory units for each server are normalized w.r.t. the servers of the first class.}
\label{fig:num_exmp}
\end{figure}
\vspace{-1mm}


The number of tasks that each user may run when monopolizing each class of servers are given in Table \ref{table3}.
Assume that all users are active. The PS-DSF (based on RDM) and the TSF allocations in this case are given in Table \ref{table4}.
Under both allocations the servers of the first two classes (the second two classes respectively) are allocated to the first (the last) two users.
According to the PS-DSF allocation, the servers of the third class (the fourth class respectively) are entirely allocated to the third user (the fourth user), which results in maximizing the minimum VDS w.r.t. these servers.

\begin{table}
\footnotesize
\caption{The total number of tasks that each user may run when monopolizing each class of servers.}
\label{table3}
\centering
\begin{tabular}{|c|c|c|c|c|}
\hline
{\bf $\gamma_{n,i}$}   &   Class A   &   Class B   &   Class C   &   Class D \\
\hline
User 1                 &     80      &     340     &     82.5    &     55    \\
\hline
User 2                 &     40      &     170     &     41.25   &     41.25 \\
\hline
User 3                 &      0      &      0      &     82.5    &     27.5  \\
\hline
User 4                 &      0      &      0      &     27.5    &     27.5  \\
\hline
\end{tabular}
\footnotesize
\caption{The total number of tasks allocated to each user from each class of servers under PS-DSF and TSF allocations.}
\label{table4}
\centering
\begin{tabular}{|c|c|c|c|c|}
\hline
{\bf PS-DSF}   &   Class A   &   Class B   &   Class C   &   Class D\\
\hline
User 1      &    40       &      170    &      0      &      0   \\
\hline
User 2      &    20       &      85     &      0      &      0    \\
\hline
User 3      &    0        &      0      &      82.5   &      0    \\
\hline
User 4      &    0        &      0      &       0     &      27.5 \\
\hline\noalign{\medskip}
\hline
{\bf TSF}   &   Class A   &   Class B   &   Class C   &   Class D\\
\hline
User 1      &    35       &      170    &      0      &      0   \\
\hline
User 2      &    22.5     &      85     &      0      &      0    \\
\hline
User 3      &    0        &      0      &      58.33   &      0    \\
\hline
User 4      &    0        &      0      &      8.05     &      27.5 \\
\hline
\end{tabular}
\end{table}

Intuitively, PS-DSF tries to allocate each server to the most efficient users.
Therefore, we expect that PS-DSF results in greater utilization for different resources of a server
compared to other allocation mechanisms such as TSF and C-DRFH.
To observe this, we have executed these algorithms over the interval $(0,300)$ sec for the cluster in Figure \ref{fig:num_exmp}.
For the PS-DSF, we start with an initial allocation and update the allocation every second according to the \emph{servers' procedure} (see our discussions in Section \ref{sec:main:DI} on distributed implementation). For TSF and C-DRFH mechanisms we precisely find these allocations every second.

It is assumed that all users except User 4 are continuously active during the simulation interval. User 4 is inactive during interval $(100,250)$ sec,
and is active elsewhere. The utilization that is achieved under any of these allocation mechanisms for the CPU at the third and the fourth classes of servers are shown respectively in Figure \ref{fig:NR1} (The CPU on the first two classes of servers and also the memory on all servers are fully utilized under any of the allocation mechanisms). It can be observed that the PS-DSF allocation mechanism results in greater utilization compared to the two other mechanisms in this example. Furthermore, it may be observed that the distributed version of the PS-DSF allocation algorithm promptly converges when changes occur in the set of active users.

\begin{figure}[h!]
\centering
\includegraphics[width=2.75in]{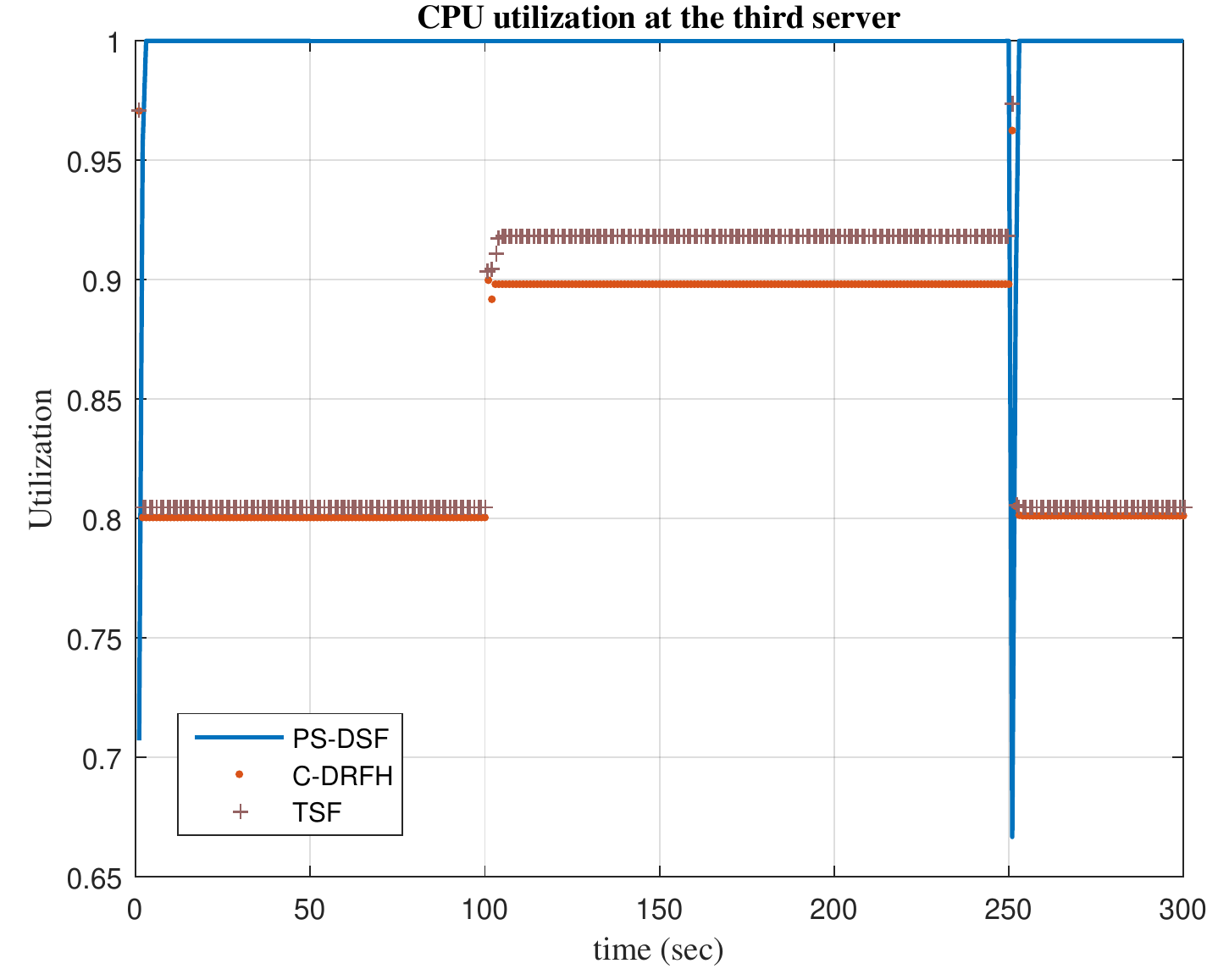}
\includegraphics[width=2.75in]{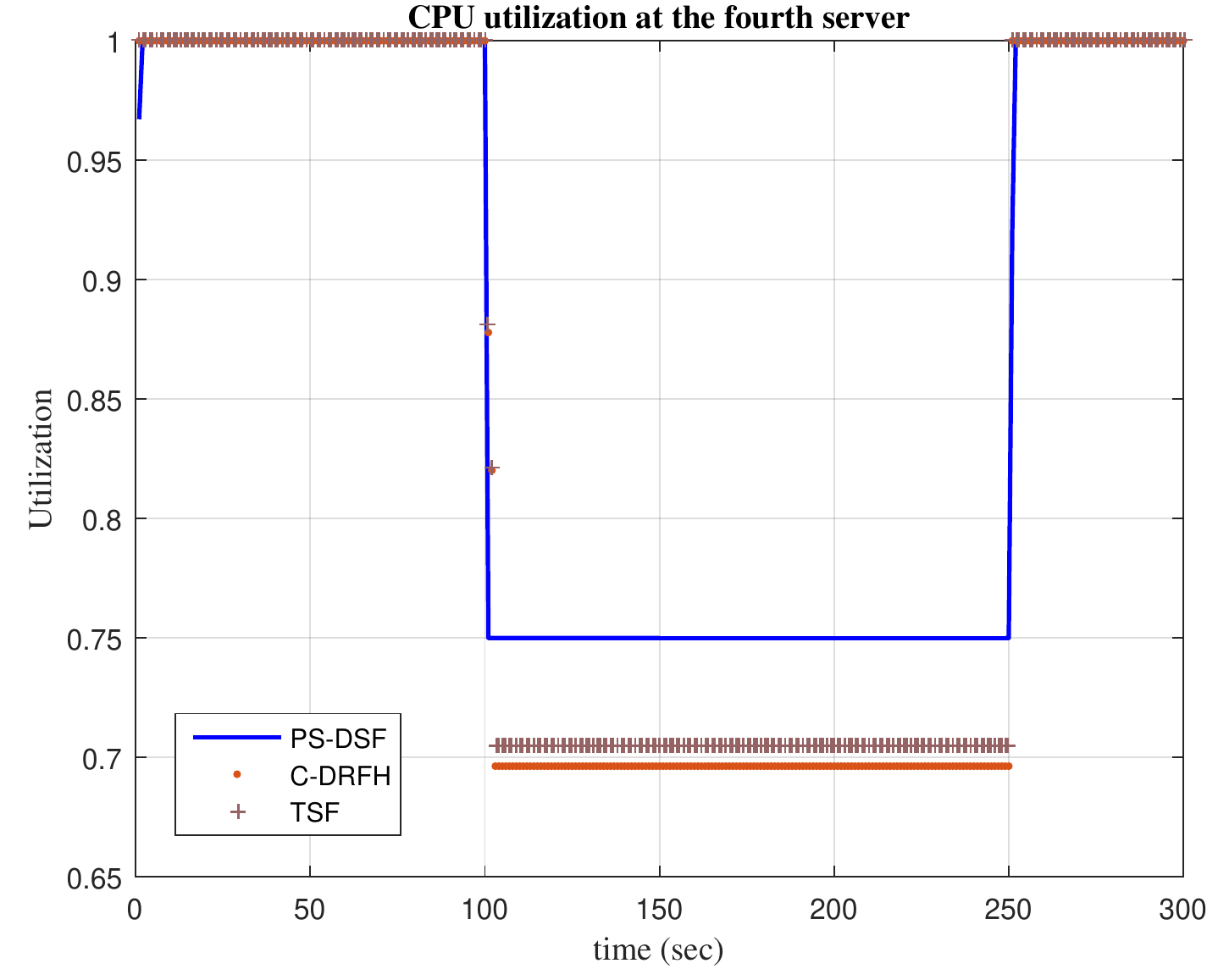}
\caption{\footnotesize The utilization that is achieved for the CPU at the third and the fourth classes of servers under PS-DSF, TSF and C-DRFH allocation mechanisms.}
\label{fig:NR1}
\end{figure}

\section{Conclusion}
In summary, we studied the problem of multi-resource fair allocation for heterogeneous servers while respecting placement constraints.
We identified important shortcomings in existing multi-resource fair allocation mechanisms when used in such environments.
Hence, we proposed a new allocation mechanism, called PS-DSF.
We discussed how our proposed allocation mechanism achieves different sharing properties which are satisfied under DRF in the case of one resource-pool/server. 
Furthermore, we discussed how PS-DSF could be implemented in a distributed manner. The performance of the PS-DSF allocation mechanism was compared against the existing allocation mechanisms and its enhanced performance was demonstrated through the numerical experiments. Further studies are under way and they will appear in future work.

\bibliographystyle{IEEEtran}
\bibliography{./PS-DSF_extended_revised}

\begin{thebibliography}{10}
\providecommand{\url}[1]{#1}
\csname url@samestyle\endcsname
\providecommand{\newblock}{\relax}
\providecommand{\bibinfo}[2]{#2}
\providecommand{\BIBentrySTDinterwordspacing}{\spaceskip=0pt\relax}
\providecommand{\BIBentryALTinterwordstretchfactor}{4}
\providecommand{\BIBentryALTinterwordspacing}{\spaceskip=\fontdimen2\font plus
\BIBentryALTinterwordstretchfactor\fontdimen3\font minus
  \fontdimen4\font\relax}
\providecommand{\BIBforeignlanguage}[2]{{%
\expandafter\ifx\csname l@#1\endcsname\relax
\typeout{** WARNING: IEEEtran.bst: No hyphenation pattern has been}%
\typeout{** loaded for the language `#1'. Using the pattern for}%
\typeout{** the default language instead.}%
\else
\language=\csname l@#1\endcsname
\fi
#2}}
\providecommand{\BIBdecl}{\relax}
\BIBdecl

\bibitem{Mod_const}
V.~Chudnovsky, R.~Rifaat, J.~Hellerstein, B.~Sharma, and C.~Das, ``Modeling and
  synthesizing task placement constraints in google compute clusters,'' in
  \emph{Symposium on Cloud Computing}, 2011.

\bibitem{Ghodsi13}
A.~Ghodsi, M.~Zaharia, S.~Shenker, and I.~Stoica, ``Choosy: Max-min fair
  sharing for datacenter jobs with constraints,'' in \emph{Proc. ACM EuroSys},
  2013, pp. 365--378.

\bibitem{DRF}
A.~Ghodsi, M.~Zaharia, B.~Hindman, A.~Konwinski, S.~Shenker, and I.~Stoica,
  ``Dominant resource fairness: Fair allocation of multiple resource types,''
  in \emph{Proc. NSDI}, June 2011.

\bibitem{CDRF}
E.~Friedman, A.~Ghodsi, and C.-A. Psomas, ``Strategyproof allocation of
  discrete jobs on multiple machines,'' in \emph{Proceedings of the ACM
  conference on Economics and computation}.\hskip 1em plus 0.5em minus
  0.4em\relax ACM, 2014, pp. 529--546.

\bibitem{Chiang12}
C.~Joe-Wong, S.~Sen, T.~Lan, and M.~Chiang, ``{Multi-resource allocation:
  Fairness-efficiency tradeoffs in a unifying framework},'' \emph{IEEE/ACM
  Trans. Networking}, vol.~21, no.~6, Dec. 2013.

\bibitem{HUG}
M.~Chowdhury, Z.~Liu, A.~Ghodsi, and I.~Stoica, ``Hug: Multi-resource fairness
  for correlated and elastic demands,'' in \emph{Proc. NSDI}, Mar 2016.

\bibitem{DRFH15}
W.~Wang, B.~Liang, and B.~Li, ``Multi-resource fair allocation in heterogeneous
  cloud computing systems,'' \emph{IEEE Transactions on Parallel and
  Distributed Systems}, vol.~26, no.~10, pp. 2822--2835, Oct 2015.

\bibitem{midrr}
K.~Yap, T.~Huang, Y.~Yiakoumis, S.~Chinchali, N.~McKeown, and S.~Katti,
  ``Scheduling packets over multiple interfaces while respecting user
  preferences,'' in \emph{Proc. ACM coNEXT}, Dec. 2013.

\bibitem{CM4FQ}
\BIBentryALTinterwordspacing
J.~Khamse{-}Ashari, I.~Lambadaris, and Y.~Q. Zhao, ``Constrained multi-user
  multi-server max-min fair queuing,'' \emph{http://arxiv.org/abs/1601.04749},
  2016. [Online]. Available: \url{http://arxiv.org/abs/1601.04749}
\BIBentrySTDinterwordspacing

\bibitem{Ashari16a}
J.~Khamse-Ashari, G.~Kesidis, I.~Lambadaris, B.~Urgaonkar, and Y.~Zhao,
  ``Max-min fair scheduling of variable-length packet-flows to multiple servers
  by deficit round-robin,'' in \emph{Proceedings of the CISS, Princton}, Mar
  2016.

\bibitem{Ashari16b}
------, ``Constrained max-min fair scheduling of variable-length packet-flows
  to multiple servers,'' in \emph{Proc. IEEE Globecom}, Dec 2016.

\bibitem{Ashari17}
------, ``{Efficient and Fair Scheduling of Placement Constrained Threads on
  Heterogeneous Multi-Processors},'' in \emph{preprint}, Sept. 2016.

\bibitem{UDRF}
Y.~Tahir, S.~Yang, A.~Koliousis, and J.~McCann, ``Udrf: Multi-resource fairness
  for complex jobs with placement constraints,'' in \emph{GLOBECOM}, Dec 2015,
  pp. 1--7.

\bibitem{TSF}
W.~Wang, B.~Li, B.~Liang, and J.~Li, ``Multi-resource fair sharing for
  datacenter jobs with placement constraints,'' \emph{SC 2016}.

\bibitem{BLi16}
------, ``Towards multi-resource fair allocation with placement constraints,''
  in \emph{Proc. ACM SIGMETRICS}, Antibes, France, 2016.

\bibitem{scheduling-TR2}
G.~Kesidis, Y.~Wang, B.~Urgaonkar, J.~Khamse-Ashari, and I.~Lambadaris, ``{Fair
  Scheduling of Multiple Resource Types over Multiple and Heterogeneous
  Resource Pools},'' CSE Dept, PSU, Tech. Rep. CSE-16-009, Sept. 08, 2016,
  http://www.cse.psu.edu/research/publications/tech-reports/2016/CSE-16-009.pdf.

\bibitem{DRF12}
D.~Parkes, A.~Procaccia, and N.~Shah, ``Beyond dominant resource fairness:
  Extensions, limitations, and indivisibilities,'' in \emph{Proc. ACM EC},
  Valencia, Spain, June 2012.

\bibitem{Google11}
C.~Reiss, J.~Wilkes, and J.~L. Hellerstein, ``Google cluster-usage traces,''
  2011, \url{http://code.google.com/p/googleclusterdata/}.

\end{thebibliography}

\appendix
\begin{proof}[Proof of Theorem~\ref{NS_condition}]
First consider a feasible allocation $\{x_{m,i}\}$ for which there exists a bottleneck resource for every user w.r.t. every eligible server.
Let $b(n,i)$ denote the bottleneck resource for user $n$ w.r.t. server $i$. That is, $b(n,i)$ is saturated, $d_{n,b(n,i)}>0$, and
\be
\frac{s_{n,i}}{\phi_n}\ge\frac{s_{m,i}}{\phi_m},~\forall m\mbox{ such that }x_{m,i}d_{m,b(n,i)}>0.\label{PTH1_1}
\ee
Given that $b(n,i)$ is saturated, it is not possible to increase $x_{n,i}$, unless decreasing $x_{m,i}$ for some user $m$ with $x_{m,i}d_{m,b(n,i)}>0$. On the other hand, \eqref{PTH1_1} implies that ${s_{m,i}}/{\phi_m}\le{s_{n,i}}/{\phi_n}$ for any user $m$ with $x_{m,i}d_{m,b(n,i)}>0$. Hence, we may not increase the allocated tasks to user $n$ from any server $i$ unless decreasing $x_{m,i}$ for some user $m$ with ${s_{m,i}}/{\phi_m}\le{s_{n,i}}/{\phi_n}$. This implies that $\{x_{m,i}\}$ satisfies PS-DSF.

Now consider an allocation $\{x_{m,j}\}$ which satisfies PS-DSF.
Let $\mathcal{R}_{n,i}$ denote the set of demanded resources by user $n$ which are saturated at an eligible server $i$ under the allocation $\{x_{m,j}\}$, that is:
\be
\mathcal{R}_{n,i}:=\{r\mid d_{n,r}>0\mbox{ and }\sum_m x_{m,i}d_{m,r}=c_{i,r}\}.
\ee
This set includes the \emph{potential} bottleneck resources for user $n$ w.r.t. server $i$.
First we prove that $\mathcal{R}_{n,i}$ may not be empty under a PS-DSF allocation.
By contradiction, assume that $\mathcal{R}_{n,i}=\emptyset$, that is none of the demanded resources by user $n$ are saturated at server $i$.
In this case, we may increase $x_{n,i}$ by:
\be
z_{n,i}:=\min_{r:d_{n,r}>0}\frac{c_{i,r}-\sum_m x_{m,i}d_{m,r}}{d_{n,r}}>0,
\ee
without decreasing $x_{m,i}$ for any user $m$. However, this contradicts to the fact that $\{x_{m,j}\}$ satisfies PS-DSF.

Next, we show that there exists some resource $r\in\mathcal{R}_{n,i}$ which serves as a bottleneck for user $n$ w.r.t. server $i$.
By contradiction, assume that none of the resources in $\mathcal{R}_{n,i}$ is a bottleneck. That is, for any resource $r\in\mathcal{R}_{n,i}$
we may find some user $p$ with $x_{p,i}d_{p,r}>0$ such that ${s_{n,i}}/{\phi_n}<{s_{p,i}}/{\phi_p}$. Hence, we can increase $x_{n,i}$ by decreasing $x_{p,i}$ for some user(s) $p$ with ${s_{n,i}}/{\phi_n}<{s_{p,i}}/{\phi_p}$. That is, $x_{n,i}$ could be increased without decreasing $x_{m,i}$ for any user $m$ with ${s_{m,i}}/{\phi_m}\le{s_{n,i}}/{\phi_n}$. However, this contradicts to the fact that $\{x_{m,j}\}$ satisfies PS-DSF.
\end{proof}

\begin{proof}[Proof of Theorem~\ref{NS_condition2}]
Consider a feasible allocation $\{x_{n,i}\}$ for which \eqref{FC2_1} holds with equality, and the condition in \eqref{Per_server_cond2} is established.
Since \eqref{FC2_1} holds with equality, it is not possible to increase $x_{n,i}$,
unless decreasing $x_{m,i}$ for some user $m$ with $x_{m,i}>0$.
On the other hand, \eqref{Per_server_cond2} implies that $s_{m,i}/\phi_m\le s_{n,i}/\phi_n$
for any user $m$ with $x_{m,i}> 0$. Therefore, we may not increase the allocated tasks to any user $n$ from any server $i$
unless decreasing $x_{m,i}$ for some user $m$ with $s_{m,i}/\phi_m\le s_{n,i}/\phi_n$.
This implies that $\{x_{n,i}\}$ satisfies PS-DSF.

Now assume that $\{x_{n,i}\}$ satisfies PS-DSF.
By contradiction, assume that \eqref{FC2_1} holds with inequality for some server $i$.
In this case we may increase $x_{n,i}$ by:
\be
z_{n,i}=\gamma_{n,i}[1-\sum_m\frac{x_{m,i}}{\gamma_{m,i}}]>0,
\ee
without decreasing $x_{m,i}$ for any user $m$. However, this contradicts to the fact that $\{x_{n,i}\}$ satisfies PS-DSF.
Next, we show that \eqref{Per_server_cond2} is established under the PS-DSF allocation.
By contradiction, assume that we can find some user $p$ with $x_{p,i}>0$ such that ${s_{p,i}}/{\phi_p}>{s_{n,i}}/{\phi_n}$. Hence, we can increase $x_{n,i}$ by decreasing $x_{p,i}$ for user $p$ with ${s_{p,i}}/{\phi_p}>{s_{n,i}}/{\phi_n}$. That is, $x_{n,i}$ could be increased without decreasing $x_{m,i}$ for any user $m$ with ${s_{m,i}}/{\phi_m}\le {s_{n,i}}/{\phi_n}$. However, this contradicts to the fact that $\{x_{n,i}\}$ satisfies PS-DSF.
\end{proof}

\begin{proof}[Proof of Theorem~\ref{Properties}]
We prove single resource fairness, bottleneck fairness, envy freeness,
and sharing incentive properties for the more complicated case of RDM.
The proofs of these properties follow the same line of arguments in case of TDM, so we do not repeat them here.

\noindent{\bf Single resource fairness:}
When there is only one type of resource, then ${\bf d}_n=d_{n,1},~\forall n$ and $\gamma_{n,i}=\delta_{n,i}c_{i,1}/d_{n,1},~\forall n,i$. As a result:
\be
s_{n,i}=\frac{x_n}{\gamma_{n,i}}=\frac{x_nd_{n,1}}{c_{i,1}\delta_{n,i}}=\frac{a_n}{c_{i,1}\delta_{n,i}},
\ee
where $a_n$ is the allocated resource to user $n$ from all servers. According to the PS-DSF allocation, we may not increase $x_n$ (or equivalently $a_n$) while maintaining feasibility without decreasing $x_{m,i}$ for some user $m$ with $s_{m,i}/\phi_m\le s_{n,i}/\phi_n$ (or $a_m/\phi_m\le a_n/\phi_n$).
Therefore, the allocated resource to different users, $\{a_n\}$ satisfies (constrained) weighted max-min fairness.

\noindent{\bf Bottleneck fairness:}
Assume that there is one resource, say $r^*$, which is dominantly requested by every user from every eligible server.
By definition, $r^*$ is considered as the \emph{dominant resource} for every user $n$ w.r.t. every eligible server $i$.
Accordingly, $\gamma_{n,i}$ is given by $\gamma_{n,i}=\delta_{n,i}c_{i,r^*}/d_{n,r^*}$, and the VDS for user $n$
w.r.t. server $i$ is given by:
\be
s_{n,i}=\frac{x_n}{\gamma_{n,i}}=\frac{x_nd_{n,r^*}}{c_{i,r^*}\delta_{n,i}}=\frac{a_{n,r^*}}{c_{i,r^*}\delta_{n,i}}.
\ee
where $a_{n,r^*}$ is the amount of the bottleneck resource allocated to user $n$ from all servers.
According to the PS-DSF allocation, we may not increase $x_n$ (or equivalently $a_{n,r^*}$) while maintaining feasibility without decreasing $x_{m,i}$ for some user $m$ with $s_{m,i}/\phi_m\le s_{n,i}/\phi_n$ (or $a_{m,r^*}/\phi_m\le a_{n,r^*}/\phi_n$). Hence, the allocated bottleneck resource to different users, $\{a_{n,r^*}\}$ satisfies (constrained) weighted max-min fairness.

\noindent{\bf Envy freeness:}
Given $U_n({\bf a})$ as the utility function for user $n$, $U_n({\bf a}_{m}\phi_n/\phi_m)$ gives the utility of user $n$ of the allocated resources to user $m$ (i.e., ${\bf a}_{m}=x_{m}{\bf d}_{m}$), when adjusted according to their weights. According to \eqref{utility}:
\be\label{EF1}
U_n(\frac{\phi_n}{\phi_m}{\bf a}_{m})=\frac{\phi_n}{\phi_m}\min_r\frac{a_{m,r}}{d_{n,r}}=\frac{\phi_nx_{m}}{\phi_m}\min_r\frac{d_{m,r}}{d_{n,r}}.
\ee
We show that: 
\begin{eqnarray}\label{EF2}
\min_r\frac{d_{m,r}}{d_{n,r}} \le \frac{d_{m,\rho(n,i)}}{d_{n,\rho(n,i)}}
                               =  \frac{\gamma_{n,i}d_{m,\rho(n,i)}}{c_{i,\rho(n,i)}}
                              \le \frac{\gamma_{n,i}}{\gamma_{m,i}}, ~\forall i.
\end{eqnarray}
As a result:
\be\label{EF3}
U_n(\frac{\phi_n}{\phi_m}{\bf a}_{m})\le {x_{m}}\frac{\phi_n\gamma_{n,i}}{\phi_m\gamma_{m,i}}, ~\forall i.
\ee
According to Theorem 1, there exists a bottleneck resource for every user w.r.t. every eligible server. Consider server $i$ for which $x_{m,i}>0$.
Let $b(n,i)$ denote the bottleneck resource for user $n$ w.r.t. server $i$. For $U_n({\bf a}_{m}\phi_n/\phi_m)$ to be greater than zero (see \eqref{EF1}), we need $d_{m,b(n,i)}>0$. Given that $b(n,i)$ is the bottleneck for user $n$ and $x_{m,i}d_{m,b(n,i)}>0$, it follows that:
\be
\frac{x_m}{\phi_m\gamma_{m,i}}\le\frac{x_n}{\phi_n\gamma_{n,i}}.
\ee
This along with \eqref{EF3} results in:
\be
U_n(\frac{\phi_n}{\phi_m}{\bf a}_{m})\le x_n.
\ee

\noindent{\bf Sharing Incentive:}
Without loss of generality assume that the demand vector for every user $n$, ${\bf d}_n$, is normalized by $\sum_i\gamma_{n,i}$ (the number of tasks which could be executed by user $n$ if the whole system is allocated to it). In this case, $x_n$ and $\gamma_{n,i}$ will be normalized by the same factor and the VDS for user $n$ w.r.t. different servers and also the resulting PS-DSF allocation won't be changed. Specifically, define:
\begin{eqnarray}
\hat{x}_n:=\frac{x_n}{\sum_j{\gamma_{n,j}}}\\
\hat{\gamma}_{n,i}:=\frac{\gamma_{n,i}}{\sum_j{\gamma_{n,j}}}
\end{eqnarray}
For the uniform allocation:
\begin{eqnarray}\nonumber 
 \hat{x}_n^{{unif}} &=& \frac{\phi_n}{\sum_m\phi_m}\sum_i\hat{\gamma}_{n,i}=\frac{\phi_n}{\sum_m\phi_m},
\end{eqnarray}
where the second equality follows from the fact that $\sum_i\hat{\gamma}_{n,i}=1$.
We assert that $\hat{x}_n/\phi_n$ is greater than or equal to $1/{\sum_m\phi_m}$ for all users under the PS-DSF allocation.

The proof is by induction on the number of users, $N$. Specifically, for $N=2$ consider two users, $n$ and $m$. 
To consider the worst-case, assume that both users have the same bottleneck w.r.t. each server.
Assume that servers are indexed in increasing order of $\hat{\gamma}_{n,j}/\hat{\gamma}_{m,j}$.
Let $j_0$ denote the least indexed server from which some tasks are allocated to user $n$ under the PS-DSF allocation, that is $\hat{x}_{n,j_0}>0$.
Given that $x_{n,j_0}>0$ and both users have the same bottleneck w.r.t. server $j_0$, it follows that (see Definition~\ref{Def_bottleneck}):
\be
s_{m,j_0}/\phi_m\ge s_{n,j_0}/\phi_n.\label{SI4}
\ee

Without loss of generality assume that $\hat{\gamma}_{n,j}/\hat{\gamma}_{m,j}>\hat{\gamma}_{n,j_0}/\hat{\gamma}_{m,j_0}$, for $j> j_0$.
For these servers it follows that $s_{m,j}/\phi_m > s_{n,j}/\phi_n$.
This along with the assumption that user $m$ has the same bottleneck as user $n$ imply that $\hat{x}_{m,j}=0$ for $j>j_0$.
Therefore, server $j_0$ is the only server for which $\hat{x}_{n,j}\hat{x}_{m,j}$ could be greater than zero. Accordingly:
\begin{eqnarray}
\hat{x}_n=\alpha_{n,j_0}\hat{\gamma}_{n,j_0} + \sum_{j=j_0+1}^K\hat{\gamma}_{n,j},\label{SI5}\\
\hat{x}_m=\alpha_{m,j_0}\hat{\gamma}_{m,j_0} + \sum_{j=1}^{j_0-1}\hat{\gamma}_{m,j},\label{SI6}
\end{eqnarray}
where $\alpha_{n,j_0}$ ($\alpha_{m,j_0}$ respectively) denotes the portion of the DR for user $n$ (user $m$) w.r.t. server $j_0$
that is allocated to it under PS-DSF.
Substituting $\hat{x}_n$ and $\hat{x}_m$ from \eqref{SI5} and \eqref{SI6} into \eqref{SI4} results in:
\begin{eqnarray}\nonumber
\frac{\alpha_{m,j_0}}{\phi_m} + \sum_{j=1}^{j_0-1}\frac{\hat{\gamma}_{m,j}}{\phi_m\hat{\gamma}_{m,j_0}}
 & \ge & \frac{\alpha_{n,j_0}}{\phi_n} + \sum_{j=j_0+1}^K\frac{\hat{\gamma}_{n,j}}{\phi_n\hat{\gamma}_{n,j_0}}\\ \nonumber
 & \ge & \frac{1-\alpha_{m,j_0}}{\phi_n} + \sum_{j=j_0+1}^K\frac{\hat{\gamma}_{n,j}}{\phi_n\hat{\gamma}_{n,j_0}},
\end{eqnarray}
where the second inequality follows from the fact that $\alpha_{m,j_0}+\alpha_{n,j_0}\ge1$.
After some manipulations, it follows that $\alpha_{m,j_0}\ge(A+\phi_m)/(\phi_m+\phi_n)$, where:
\begin{eqnarray}
A:= \sum_{j=j_0+1}^K\phi_m\frac{\hat{\gamma}_{n,j}}{\hat{\gamma}_{n,j_0}} - \sum_{j=1}^{j_0-1}\phi_n\frac{\hat{\gamma}_{m,j}}{\hat{\gamma}_{m,j_0}}.
\end{eqnarray}
Applying the lower bound of $\alpha_{m,j_0}$ into \eqref{SI6} and after some manipulations, it follows that:
\begin{eqnarray}\nonumber
\frac{\hat{x}_m}{\phi_m} &\ge& \frac{\sum_{j=1}^{j_0}{\hat{\gamma}_{m,j}}+
                          \sum_{j=j_0+1}^K\frac{\hat{\gamma}_{n,j}}{\hat{\gamma}_{n,j_0}}\hat{\gamma}_{m,j_0}}{\phi_m+\phi_n}\\
                         &\ge& \frac{\sum_{j=1}^{K}{\hat{\gamma}_{m,j}}}{\phi_m+\phi_n} = \frac{1}{\phi_m+\phi_n}\label{SI8}
\end{eqnarray}
where the second inequality follows from the fact that $\hat{\gamma}_{n,j}/\hat{\gamma}_{m,j}\ge\hat{\gamma}_{n,j_0}/\hat{\gamma}_{m,j_0},~j\ge j_0$,
and the last equality follows from the fact that $\sum_j\hat{\gamma}_{m,j}=1$. The lower bound in \eqref{SI8} could be shown
for ${\hat{x}_n}/{\phi_n}$ in the same way.

Assume that the statement is established for $N=N_0$ users.
For the case that the set of all users, $\mathcal{N}$, consists of $N_0+1$ users, we may assume that $\mathcal{N}$
is comprised of a subset $\mathcal{N}_0$ of $N_0$ users with the total weight of $\Phi_0:=\sum_{n\in\mathcal{N}_0}\phi_n$
and a singular user $n_0$ with the weight of $\phi_0$. We assume that user $n_0$ is chosen arbitrarily.
We may consider $\phi_n/(\Phi_0+\phi_0)$ portion of the resources on every server as the share of each user $n$.
Assume that user $n_0$ does not share its resources with others. In this case, $\hat{x}_{n_0}={\phi_0}/(\Phi_0+\phi_0)$.

User $n_0$ would prefer to exchange all or part of its allocated resources from server $j$ with all or part of the allocated resources to user $m$ from server $i$, if
\be
\frac{\hat{\gamma}_{n_0,i}}{\hat{\gamma}_{m,i}}>\frac{\hat{\gamma}_{n_0,j}}{\hat{\gamma}_{m,j}}.
\ee
In this case, the number of allocated tasks to both of them could be increased compared to the generic uniform allocation.
We may repeat the same process, exchanging the allocated resources to user $n_0$ by that for other users,
until no more exchange is possible and $\hat{x}_{n_0}$ cannot be further increased.
After that, we may freeze the allocated resources to user $n_0$ and allocate the remaining resources among other users according to PS-DSF.
Given that sharing incentive is provided by PS-DSF for the set $\mathcal{N}_0$ with $N_0$ users, it follows that:
\be\label{N0_inequality}
{\hat{x}_n}\ge\frac{\Phi_0}{\Phi_0+\phi_0}\frac{{\phi_n}}{\Phi_0}=\frac{{\phi_n}}{\Phi_0+\phi_0},~\forall n\in\mathcal{N}_0.
\ee

If we allocate the whole resources among all users $n\in\mathcal{N}$ according to PS-DSF allocation, the number of allocated tasks to users
$n\in\mathcal{N}_0$ may not be decreased compared to the above-described allocation (because there is no reservation for user $n_0$ in this case).
That is, \eqref{N0_inequality} is established under the PS-DSF allocation. Since $n_0$ is chosen arbitrarily, we may repeat the same discussions
by choosing a different set $\mathcal{N}'_0$ which includes $n_0$. Hence, we may conclude that the lower bound in \eqref{N0_inequality} is established
for all users.

\noindent{\bf Pareto optimality:}
Consider an allocation, $\{x_{n,i}\}$, satisfying PS-DSF based on TDM.
For such an allocation, Theorem~\ref{NS_condition2} implies that \eqref{FC2_1} holds with equality for each server $i$.
Hence, we may not increase $x_{n,i}$ without decreasing $x_{m,i}$ for some user $m$ with $x_{m,i}>0$.
Furthermore, according to Theorem~\ref{NS_condition2}, for any user $m$ and server $i$ with $x_{m,i}>0$:
\be
x_{m,i}>0~\Rightarrow~s_{m,i}/\phi_m=\min_n s_{n,i}/\phi_n.\label{maximization_constraint}
\ee
In fact, PS-DSF allocation mechanism maximizes $x_m$ for each user $m$ subject to \eqref{maximization_constraint}.
The following lemma shows that this condition is not restricting, as we may not increase $x_n$ without decreasing $x_{m}$ for some user $m$,
even when violating the condition in \eqref{maximization_constraint}. This means that PS-DSF allocation is Pareto optimal in case of TDM.
\begin{lemma}
Assume that $\{x_{n,i}\}$ satisfies PS-DSF based on TDM.
Consider two arbitrary users, $n$ and $m$, for which $x_{n,i}>0$ and $x_{m,j}>0$.
If user $n$ exchanges all or part of its allocated tasks from server $i$ with all or part of the allocated tasks to user $m$ from server $j$, then the allocated tasks to at least one of them will be decreased compared to the PS-DSF allocation.
\end{lemma}
\begin{proof}
Given that $x_{m,j}>0$ and $x_{n,i}>0$, we may decrease $x_{m,j}$ and $x_{n,i}$, and increase $x_{n,j}$ and $x_{m,i}$. Let $\Delta x_{n,i}$, $\Delta x_{m,i}$ and $\Delta x_{n,j}$, $\Delta x_{m,j}$ denote a \emph{feasible} change in the number of allocated tasks to users $n$ and $m$ while \eqref{FC2_1} holds with equality for both servers $i$ and $j$.
For \eqref{FC2_1} to hold with equality we have:
\be
\Delta x_{n,j}=-\Delta x_{m,j}\frac{\gamma_{n,j}}{\gamma_{m,j}}\label{POP1}\\
\Delta x_{m,i}=-\Delta x_{n,i}\frac{\gamma_{m,i}}{\gamma_{n,i}}\label{POP2}
\ee
Assume that $\Delta x_{n,i}+\Delta x_{n,j}>0$ or $-\Delta x_{n,i}<\Delta x_{n,j}$. This along with \eqref{POP1} and \eqref{POP2} results in:
\be
\Delta x_{m,i} & < & \Delta x_{n,j}\frac{\gamma_{m,i}}{\gamma_{n,i}}\\
               & < & -\Delta x_{m,j}\frac{\gamma_{n,j}}{\gamma_{m,j}}\frac{\gamma_{m,i}}{\gamma_{n,i}}.\label{POP4}
\ee
The fact that $x_{m,j}>0$ and $x_{n,i}>0$ along with \eqref{Per_server_cond2} result in:
\be
\frac{\gamma_{m,i}}{\gamma_{n,i}}\le\frac{\phi_n}{\phi_m}\frac{x_m}{x_n}\label{POP5}\\
\frac{\gamma_{n,j}}{\gamma_{m,j}}\le\frac{\phi_m}{\phi_n}\frac{x_n}{x_m}\label{POP6}
\ee
Combining \eqref{POP4} with \eqref{POP5} and \eqref{POP6} results in $\Delta x_{m,i}<-\Delta x_{m,j}$.
This means that the number of allocated tasks to user $m$ is decreased compared to PS-DSF allocation.
\end{proof}
%

\noindent{\bf Strategy proofness:}
Let $A:=\{a_{m,i}\}$ (and $A':=\{a'_{m,i}\}$, respectively) denote the resulting PS-DSF allocation when user $n$ trustfully declares ${\bf d}_n$ and
${\bf \delta}_n=[\delta_{n,i}]$ (non-trustfully declares ${\bf d}_n'$ and ${\bf \delta}'_n$). Users other than $n$ take the same actions in both cases
(whether trustful or non-trustful). Hence, $\gamma'_{m,i}=\gamma_{m,i}$ for $m\neq n$.
The number of tasks that user $n$ may actually execute under the allocation $A'$ (i.e., by using ${\bf a}'_n=x'_n{\bf d}'_n$) is given by:
\be\label{SP3}
U_n({\bf a}'_n)=\min_r\frac{a'_{n,r}}{d_{n,r}}=x'_n\min_r\frac{d'_{n,r}}{d_{n,r}}.
\ee
As in \eqref{EF2}, we can show that $\min_r\frac{d'_{n,r}}{d_{n,r}}\le\frac{\gamma_{n,i}}{\gamma'_{n,i}},~\forall i$.
Hence:
\be
U_n({\bf a}'_n)=x'_n\min_r\frac{d'_{n,r}}{d_{n,r}}\le x'_n\frac{\gamma_{n,i}}{\gamma'_{n,i}}.
\ee

For strategy proofness we need to show that $U_n({\bf a}'_n)\le x_n$.
By contradiction, assume that $U_n({\bf a}'_n)>x_n$. It follows that:
\be
s_{n,i}=\frac{x_n}{\gamma_{n,i}}<\frac{U_n(a'_n)}{\gamma_{n,i}}\le\frac{x'_n}{\gamma'_{n,i}}=s'_{n,i},~\forall i.
\ee
That is, the VDS for user $n$ is increased w.r.t. all servers under the allocation $A'$
compared to the allocation $A$, provided that $U_n({\bf a}'_n)>x_n$.
Let $\mathcal{U}$ denote the set of users for which $U_m({\bf a}'_m)>x_m$. For users $m\in\mathcal{U}$, $m\neq n$, it follows that
\be
s_{m,i}=\frac{x_m}{\gamma_{m,i}}<\frac{U_m(a'_m)}{\gamma_{m,i}}=\frac{x'_m}{\gamma'_{m,i}}=s'_{m,i},~\forall i.
\ee
We define:
\be
s_i:=\min_n\frac{s_{n,i}}{\phi_n} \label{VDSL}
\ee
as the Virtual Dominant Share Level, VDSL, at server $i$ under the allocation $A$.
In the same way, we define $s'_i$ as the VDSL at server $i$ under the allocation $A'$.
For any user $m\in\mathcal{U}$, Theorem~\ref{NS_condition2} implies that $s'_i=s'_{m,i}/\phi_m$ provided that $x'_{m,i}>0$. It follows that:
\be
s'_i=\frac{s'_{m,i}}{\phi_m}>\frac{s_{m,i}}{\phi_m}\ge s_i.
\ee
That is, the VDSL is increased at all servers for which $x'_{m,i}>0$ for some $m\in\mathcal{U}$.
Let define:
\be
\mathcal{S}:=\{i\mid x'_{m,i}>0\mbox{ for some }m\in\mathcal{U}\}.
\ee
Accordingly, no tasks are allocated under the allocation $A'$ from servers $j\notin\mathcal{S}$ to users $m\in\mathcal{U}$,
i.e., $x'_{m,j}=0,~m\in\mathcal{U},~j\notin\mathcal{S}$. Hence, VDSL at servers $j\notin\mathcal{S}$ may not be decreased under the allocation $A'$ compared to allocation $A$, that is $s'_j\ge s_j$ for servers $j\notin\mathcal{S}$. Therefore, $s'_i\ge s_i,~\forall i$, which in turn implies that $U_m({\bf a}'_m)=x'_m\ge x_m$ for $m\neq n$. This along with the assumption that $U_n({\bf a}'_n)>x_n$ contradict to Pareto optimality of the allocation $A$.
\end{proof}

\begin{proof}[Proof of Lemma~\ref{strategy_proofness}]
The proof follows the same line of arguments as the proof of strategy proofness in case of TDM.
Specifically, 
when all users demand all types of resources, the same resource serves as the bottleneck for all users w.r.t. each server.
Hence, we may define the VDSL at each server $i$ as in \eqref{VDSL}.
With the same line of arguments we may conclude that the VDSL is not decreased at any server $i$ under the allocation $A'$ compared to allocation $A$, provided that $U_n(a'_n)\ge x_n$. That is, $s'_i\ge s_i$, which in turn implies that $U_m({\bf a}'_m)\ge x_m$ $\forall m$. Therefore, user $n$ may not decrease the utilization of other users, by lying about its resource demands or the set of eligible servers, unless decreasing its own utilization.
\end{proof}

\end{document}